\documentclass[11p,reqno]{amsart}
\usepackage[T1]{fontenc}
\usepackage{graphicx}
\usepackage{amssymb,amsthm,amsmath,mathrsfs,bm,braket,marginnote}
\usepackage{enumerate}
\usepackage{appendix}
\usepackage[colorlinks=true, pdfstartview=FitV, linkcolor=blue, citecolor=blue, urlcolor=blue]{hyperref}
\usepackage{multirow}
\usepackage{pgf}
\usepackage{pgfplots}
\usepackage{tikz}
\usetikzlibrary{arrows,calc}
\usepackage{verbatim}
\usetikzlibrary{decorations.pathreplacing,decorations.pathmorphing}
\usepackage[numbers,sort&compress]{natbib}
\usepackage{dsfont}
\usepackage{float}

\numberwithin{equation}{section}
\newtheorem{theorem}{Theorem}[section]

\newtheorem{definition}[theorem]{Definition}
\newtheorem{remark}[theorem]{Remark}
\newtheorem{coro}[theorem]{Corollary}
\newtheorem{example}[theorem]{Example}

\newtheorem{proposition}[theorem]{Proposition}

\newcommand{\bphi}{\boldsymbol \phi}
\newcommand{\bpsi}{\boldsymbol \psi}

\begin{document}

\title[A twisted Marchenko lemma and applications]{A twisted Marchenko lemma and applications}

\date{}

\dedicatory{}


\author{Jun-Yang Zhang}
\address{School of Mathematical Sciences, Zhejiang University of Technology, Hangzhou, 310023, Zhejiang, China}
\email{zhangjunyang142@gmail.com}

\author{Shi-Hao Li}
\address{Department of Mathematics, Sichuan University, Chengdu, 610064, PR China}
\email{shihao.li@scu.edu.cn}

\author{Shou-Feng Shen}
\address{School of Mathematical Sciences, Zhejiang University of Technology, Hangzhou, 310023, Zhejiang, China}
\email{mathssf@zjut.edu.cn}

\begin{abstract}
We develop a twisted Marchenko lemma with applications into non-commutative two-component integrable systems. Starting from a twisted differential equation, we introduce the corresponding two-component Wronski matrix $W_{n,m}(\boldsymbol\phi,\boldsymbol\psi)$. Under the Marchenko lemma conditions, the logarithmic derivative $(\partial W_{n,m})W_{n,m}^{-1}$ satisfies nonlinear equations, whose two-component projections yield a non-commutative two-component KP equation, while a single-component reduction recovers the non-commutative Burgers equation, KP equation and DS equation. A direct approach based on two-component quasi-Wronskians is developed, from which a higher-order non-commutative two-component KP equation is generated with matrix-valued solutions. Reductions of the seed functions of the non-commutative two-component KP hierarchy yield the one-component counterpart, while a stepwise reduction of those for the two-component KP equation gives the non-commutative DS and NLS equations.
\end{abstract}

\maketitle

\section{Introduction}

One of the central ideas in the theory of integrable systems is that nonlinear equations often possess an underlying linear structure. Although the equations of primary interest are nonlinear, their integrability is typically revealed through auxiliary linear problems. The solvability of these equations reflects the principle that, rather than treating a nonlinear evolution equation directly, one embeds it into a larger linear framework and reconstructs the nonlinear dependent variables from suitable linear spectrum and dispersion relations.

Marchenko developed an algebraic formulation of this linear-to-nonlinear mechanism. Instead of starting from a specific spectral problem, one considers an associative algebra equipped with mutually commuting derivations and an invertible element $\Gamma$ satisfying suitable linear differential relations. The nonlinear dependent variables are then introduced via logarithmic derivatives such as $(\partial \Gamma)\Gamma^{-1}$. The essence of the Marchenko lemma is that the linear relations imposed on $\Gamma$ induce nonlinear differential relations for these logarithmic derivatives. In this way, the construction isolates, in a purely algebraic form, the mechanism by which nonlinear integrable dynamics emerges from compatible linear flows \cite{marchenko88}. The following simple example illustrates this phenomenon. 
\begin{example}\label{example}
Suppose that an invertible function $\Gamma$ satisfies a linear heat equation $\Gamma_t=\Gamma_{xx}$. Then by introducing $u=\Gamma_x\Gamma^{-1}$, the compatibility of the defining relations implies a nonlinear Burgers equation $u_t=u_{xx}+2u_xu$. 
\end{example}
An important feature of Marchenko's construction is that its basic mechanism is not intrinsically tied to commutativity. Indeed, the passage from linear differential relations for an invertible element $\Gamma$ to nonlinear relations for its logarithmic derivatives is based essentially on associativity, the Leibniz rule for derivations, and the existence of inverses. It is therefore natural to ask whether the same linear-to-nonlinear mechanism persists when the dependent variables take values in a non-commutative algebra. The algebraic nature of Marchenko's argument makes it particularly well suited to non-abelian generalizations.
This perspective was systematically developed by Etingof, Gelfand and Retakh \cite{etingof98}. They studied integrable equations for functions taking values in an arbitrary unital associative algebra and combined Marchenko's lemmas with the theory of quasi-determinants. Explicit multi-soliton solutions of several non-abelian integrable systems were obtained, including periodic Toda field equations, Langmuir-type equations and Schrödinger-type equations. Their work shows that the Marchenko construction remains valid in a genuinely non-commutative setting.

Within this non-commutative framework, Wronski matrices provide a natural finite-dimensional realization of the Marchenko construction. Starting from a collection of algebra-valued functions $\{f_1,\cdots,f_N\}$ satisfying prescribed linear dispersion relations, one forms the Wronski matrix
\begin{align}\label{opwronski}
\Gamma=W(f_1,\cdots,f_N):=\begin{pmatrix}f_1 & \cdots & f_N\\\partial f_1 & \cdots & \partial f_N\\\vdots & & \vdots\\\partial^{N-1}f_1 & \cdots & \partial^{N-1}f_N\end{pmatrix},
\end{align}
where $\partial$ denotes a derivation of the underlying algebra.
In this way, the nonlinear dependent variables are encoded in only a small portion of the matrix $(\partial \Gamma)\Gamma^{-1}$. In the commutative setting, these variables reduce to ratios of Wronskian determinants, whereas in the non-commutative setting they admit natural quasi-determinant expressions. Thus, Wronski matrices and quasi-determinants provide the structural and algebraic ingredients needed to convert the abstract non-commutative Marchenko mechanism into explicit solution formulae.
\subsection{Main results}
The purpose of the present paper is to develop a twisted Marchenko lemma adapted to the non-commutative two-component KP hierarchy. 
The main difficulty in extending the usual non-commutative Marchenko framework to the two-component case is that the two families of KP flows are coupled through a common dressing structure and cannot be treated as two independent copies of the one-component theory.
To address this, we begin with a special twisted differential equation 
\begin{equation*}
\begin{aligned}
\partial^n\phi
=\sum_{i=1}^{n}\xi_i\partial^{i-1}\phi
+\sum_{j=1}^{m}\xi_{n+j}\partial^{j-1}\psi,\quad 
\partial^m\psi
=\sum_{i=1}^{n}\eta_i\partial^{i-1}\phi
+\sum_{j=1}^{m}\eta_{n+j}\partial^{j-1}\psi,
\end{aligned}
\end{equation*}
whose coefficients $\{\xi_k,\eta_k\}_{k=1}^{n+m}$ belong to an associative but not necessarily commutative algebra. The solution of this twisted differential equation yields the generalized Wronski matrix $W_{n,m}(\bphi,\bpsi)$ defined in Definition \ref{Def:bwronski}.
The invertibility of this Wronski matrix is crucial for the construction of the twisted Marchenko lemma. We show that, under the assumption that the two-component Wronski matrix is invertible, the coefficients $\{\xi_k,\eta_k\}_{k=1}^{n+m}$ can be expressed in terms of quasi-determinants (See Theorem \ref{coefficients}).
We reformulate this theorem in a matrix form and demonstrate that $(\partial W_{n,m}(\bphi,\bpsi))W_{n,m}(\bphi,\bpsi)^{-1}$ depends only on a small number of dependent variables, thereby enabling the identification of the corresponding non-commutative integrable systems.

In Section \ref{sec:3}, we establish a twisted Marchenko lemma. We assume
the existence of two families of derivatives $\partial_1^{(1)}, \partial_2^{(1)}$ and $\partial_1^{(2)}, \partial_2^{(2)}$, such that the two-component Wronski matrix $W_{n,m}(\bphi,\bpsi)$ satisfies the following linear relations
\begin{align*}
&\partial_{1}^{(1)}W_{n,m}-P_1\partial W_{n,m}=W_{n,m} A_1,\quad \partial_{2}^{(1)}W_{n,m}-P_1\partial^2W_{n,m}=W_{n,m} C_1,\\
&\partial_{1}^{(2)}W_{n,m}-P_2\partial W_{n,m}=W_{n,m} A_2,\quad \partial_{2}^{(2)}W_{n,m}-P_2\partial^2W_{n,m}=W_{n,m} C_2
\end{align*}
for constant matrices $A_\ell,C_\ell$ and $P_\ell$ for $\ell=1,2$. In Theorem \ref{2c_marchenko_lemma}, we prove that the matrix $\mathcal{Q}=(\partial W_{n,m}(\bphi,\bpsi))W_{n,m}(\bphi,\bpsi)^{-1}$ satisfies a nonlinear equation
\begin{align*}
&\partial_2^{(\ell)}\mathcal{Q}-P_\ell(\partial^2\mathcal{Q}+2(\partial \mathcal{Q})\mathcal{Q})-[P_\ell,\mathcal{Q}](\partial \mathcal{Q}+\mathcal{Q}^2)=0,\\
&\partial_1^{(\ell)}\mathcal{Q}-P_\ell \partial\mathcal{Q}-[P_\ell,\mathcal{Q}]\mathcal{Q}=0.
\end{align*}
 An equivalent coupled system is obtained in Theorem \ref{2c_marchenko_projection} by projecting the equation for $\mathcal{Q}$ onto its two distinct components. We specialize the nonlinear equations to the non-commutative two-component KP equation (Section \ref{sec:3.2}) and the DS equation (Section \ref{sec:3.3})  by imposing different linear relations on seed functions $\{\phi_j\}_{j=1}^{n+m}$ and $\{\psi_j\}_{j=1}^{n+m}$ in $W_{n,m}(\bphi,\bpsi)$.

A direct approach to the two-component system is developed in Section \ref{sec:4}.  We verify that the two-component Wronskian quasi-determinant indeed provides a solution to the non-commutative two-component KP equation. 
We show that the derivatives of generalized quasi-Wronskian determinants admit  a closed form expression, and that linear combinations of the derivatives of different orders  lead us to an equation involving only a few potential functions. In this way, we recover the non-commutative two-component KP equation via the direct method, and also uncover some higher-order non-commutative two-component KP systems. 

In Section~\ref{sec:5}, we formulate reduction techniques within the direct method.  We first specialize the seed functions of the non-commutative two-component KP hierarchy to those of its one-component counterpart.  We then perform a stepwise reduction on the seed functions for the non-commutative two-component KP equation, yielding the non-commutative DS and NLS equations alongside their associated quasi-determinant solutions.

\section{Twisted differential equation and two-component Wronski matrix}
Let $\mathcal A$ be a unital algebra and $\partial:\mathcal A \to \mathcal A$ be a derivation of $\mathcal A$. Given $2(n+m)$ coefficients $\xi_k,\eta_k\in\mathcal A$, we consider a coupled system of two differential equations of orders $n$ and $m$ respectively, such that
\begin{equation}
\label{coupled_system}
\begin{aligned}
\partial^n\phi
=\sum_{i=1}^{n}\xi_i\partial^{i-1}\phi
+\sum_{j=1}^{m}\xi_{n+j}\partial^{j-1}\psi,\qquad
\partial^m\psi
=\sum_{i=1}^{n}\eta_i\partial^{i-1}\phi
+\sum_{j=1}^{m}\eta_{n+j}\partial^{j-1}\psi. 
\end{aligned}
\end{equation}
We call $(\phi,\psi)$ a solution pair if both equations in~\eqref{coupled_system} are satisfied simultaneously. To uniquely determine the coefficients $\xi_k$ and $\eta_k$ from  solution pairs of~\eqref{coupled_system}, we introduce $n+m$ solution pairs $\{(\phi_j,\psi_j)\}_{j=1}^{n+m}$. 
\begin{definition}\label{Def:bwronski}
The two-component Wronski matrix of type $(n,m)$ associated with the solution pairs 
\begin{align*}
\boldsymbol\phi=(\phi_1,\ldots,\phi_{n+m}),\quad \boldsymbol\psi=(\psi_1,\ldots,\psi_{n+m})
\end{align*}
is defined by
\begin{align} 
\label{2c_wronski}
W_{n,m}(\boldsymbol\phi,\boldsymbol\psi)
=\begin{pmatrix}
\phi_1 & \phi_2 & \cdots & \phi_{n+m}\\
\partial\phi_1 & \partial\phi_2 & \cdots & \partial\phi_{n+m}\\
\vdots & \vdots & & \vdots\\
\partial^{n-1}\phi_1 & \partial^{n-1}\phi_2 & \cdots & \partial^{n-1}\phi_{n+m}\\
\psi_1 & \psi_2 & \cdots & \psi_{n+m}\\
\partial\psi_1 & \partial\psi_2 & \cdots & \partial\psi_{n+m}\\
\vdots & \vdots & & \vdots\\
\partial^{m-1}\psi_1 & \partial^{m-1}\psi_2 & \cdots & \partial^{m-1}\psi_{n+m}
\end{pmatrix}.
\end{align}
\end{definition}
Furthermore, we have the following non-degeneracy condition.
\begin{definition}
The family of $n+m$ solution pairs $\{\phi_j,\psi_j\}_{j=1}^{n+m}$ is called nondegenerate if its corresponding two-component Wronski matrix $W_{n,m}(\boldsymbol\phi,\boldsymbol\psi)$ is invertible. 
\end{definition}


\begin{theorem} \label{coefficients}
Suppose that $\{(\phi_j,\psi_j)\}_{j=1}^{n+m}$ is a nondegenerate family of solution pairs of~\eqref{coupled_system}, then the coefficients in~\eqref{coupled_system} are given by
\begin{align}
\label{xi_k_eta_k}
\xi_k
=-\begin{vmatrix}
W_{n,m}(\boldsymbol\phi,\boldsymbol\psi) & \boldsymbol e_k \\
\partial^n \boldsymbol\phi & \boxed{0}
\end{vmatrix}, 
\qquad
\eta_k
=-\begin{vmatrix}
W_{n,m}(\boldsymbol\phi,\boldsymbol\psi) & \boldsymbol e_k \\
\partial^m\boldsymbol\psi& \boxed{0}
\end{vmatrix}, \qquad k = 1, \ldots, n+m,
\end{align}
where $\boldsymbol e_k$ denotes the $(n+m)$-dimensional unit column vector with unit in the $k^{\mathrm{th}}$ row and 0 elsewhere.
\end{theorem}

\begin{proof}
Since $\{(\phi_j,\psi_j)\}_{j=1}^{n+m}$ are $n+m$ solution pairs of~\eqref{coupled_system}, we have the linear system
\[
\partial^n \boldsymbol\phi
=\boldsymbol\xi W_{n,m}(\boldsymbol\phi,\boldsymbol\psi),\qquad
\partial^m \boldsymbol\psi
=\boldsymbol\eta W_{n,m}(\boldsymbol\phi,\boldsymbol\psi).
\]
where $\boldsymbol\xi=\left(\xi_1,\ldots,\xi_{n+m}\right)$ and $\boldsymbol\eta=\left(\eta_1,\ldots,\eta_{n+m}\right)$.
By nondegeneracy, $W_{n,m}(\bphi,\bpsi)$ is invertible, and thus
\[
\xi_k
=(\partial^n \boldsymbol\phi) W_{n,m}(\boldsymbol\phi,\boldsymbol\psi)^{-1}\boldsymbol e_k,
\qquad
\eta_k
=(\partial^m\boldsymbol\psi) W_{n,m}(\boldsymbol\phi,\boldsymbol\psi)^{-1}\boldsymbol e_k.
\]
The formulas~\eqref{xi_k_eta_k} then follow directly from the definition of a quasi-determinant.
\end{proof}
The corresponding matrix form of the linear system is stated as follows.
\begin{coro}
\label{wronski_logarithmic_derivative} 
Set
\begin{align}\label{xww}
X=(\partial W_{n,m}(\bphi,\bpsi))W_{n,m}(\bphi,\bpsi)^{-1}=(x_{p,q})_{p,q=1}^{n+m}.
\end{align}
Then, $
x_{p,q}=\delta_{p+1,q}$ for
$p\in\{1,\ldots,n-1\}\cup\{n+1,\ldots,n+m-1\}$ and 
 $q=1,\ldots,n+m$, where $\delta_{p,q}$ denotes the Kronecker symbol, and
\[
x_{n,q}=\xi_q,\qquad
x_{n+m,q}=\eta_q
\]
for $q=1,\ldots,n+m$, 
with $\xi_q$, $\eta_q$ being the coefficients of coupled system~\eqref{coupled_system}, whose quasi-determinant expressions are given in Theorem \ref{coefficients}.
\end{coro}

\begin{proof}
Indeed, equation \eqref{xww} is equivalent to 
$$XW_{n,m}(\bphi,\bpsi)=\partial W_{n,m}(\bphi,\bpsi).$$ Thus, the $k$-th row of $X$ must satisfy
\begin{align*}
&\partial^k \phi=\sum_{i=1}^n x_{k,i}\partial^{i-1}\phi+\sum_{j=1}^m x_{k,n+j}\partial^{j-1}\psi,\qquad\,\, 1\leq k\leq n,\\
&\partial^{k-n} \psi=\sum_{i=1}^n x_{k,i}\partial^{i-1}\phi+\sum_{j=1}^m x_{k,n+j}\partial^{j-1}\psi,\quad n+1\leq k\leq n+m.
\end{align*}
The non-degeneracy of $W_{n,m}(\bphi,\bpsi)$ immediately yields the desired result. 
\end{proof}

\section{A twisted Marchenko lemma}\label{sec:3}
In this section, we take $\mathcal A$ as a unital algebra, and let
$\partial$ and $\partial_k^{(\ell)}$ $\left(k,\ell=1,2\right)$ be mutually commuting derivations of $\mathcal A$. Furthermore, we assume that elements  $A_\ell,C_\ell,P_\ell\in\mathcal A$ $\left(\ell=1,2\right)$ lie in the common kernel of all these derivations. In other words, they are constants annihilated by each of the derivations. Firstly, we show the Marchenko lemma that if $\Gamma$ satisfies coupled linear systems, then $(\partial\Gamma)\Gamma^{-1}$ satisfies a nonlinear equation.

\begin{theorem}
\label{2c_marchenko_lemma}
Let $\Gamma\in\mathcal{A}$ be an invertible element satisfying coupled linear equations
\begin{subequations}
\begin{align}
\partial_2^{(\ell)}\Gamma-P_\ell\partial^2\Gamma
&=\Gamma C_\ell,
\label{marchenko_second}\\
\partial_1^{(\ell)}\Gamma-P_\ell\partial\Gamma
&=\Gamma A_\ell.\label{marchenko_first}
\end{align}
\end{subequations}
By setting $\mathcal Q
=\left(\partial\Gamma\right)\Gamma^{-1}$, then $\mathcal Q$ satisfies the following nonlinear equations
\begin{subequations}
\begin{align}
&\partial_2^{(\ell)}\mathcal{Q}-P_\ell(\partial^2\mathcal{Q}+2(\partial \mathcal{Q})\mathcal{Q})-[P_\ell,\mathcal{Q}](\partial \mathcal{Q}+\mathcal{Q}^2)=0,\label{ncq2}\\
&\partial_1^{(\ell)}\mathcal{Q}-P_\ell \partial\mathcal{Q}-[P_\ell,\mathcal{Q}]\mathcal{Q}=0.\label{ncq1}
\end{align}
\end{subequations}
\end{theorem}

\begin{proof}
From equation~\eqref{marchenko_second} and the fact that
$C_\ell$
lies in the kernel of $\partial$, we obtain
\begin{align*}
0=\partial (\Gamma^{-1}\Gamma C_\ell)=\Gamma^{-1}\left(
-\partial \Gamma \Gamma^{-1}\partial_2^{(\ell)}\Gamma+\partial\partial_2^{(\ell)}\Gamma+\partial\Gamma\Gamma^{-1}P_\ell \partial^2\Gamma-P_\ell\partial^3\Gamma
\right),
\end{align*}
which yields
\begin{align}
-\mathcal Q\partial_2^{(\ell)}\Gamma+\partial^{(\ell)}_2\partial\Gamma+\mathcal{Q}P_\ell\partial^2\Gamma-P_\ell \partial^3\Gamma=0
\label{second_d}
\end{align}
from the invertibility of $\Gamma$. On the other hand,
by using the definition $\mathcal{Q}=(\partial \Gamma)\Gamma^{-1}$, we have
\begin{align*}
&\partial_2^{(\ell)}\mathcal Q
=(\partial_2^{(\ell)}\partial\Gamma)\Gamma^{-1}-\mathcal Q(\partial_2^{(\ell)}\Gamma)\Gamma^{-1},\\
&\partial \mathcal{Q}=(\partial^2\Gamma)\Gamma^{-1}-\mathcal{Q}^2,\\
&\partial^2\mathcal Q
=(\partial^3\Gamma)\Gamma^{-1}-2(\partial\mathcal Q)\mathcal Q-\mathcal Q\partial\mathcal Q-\mathcal Q^3.
\end{align*}
Substituting these into~\eqref{second_d} gives the equation \eqref{ncq2}.
A similar computation starting from~\eqref{marchenko_first} yields \eqref{ncq1}.
\end{proof}
In the next theorem, we show that if $P_1$ and $P_2$ are specially chosen as orthogonal idempotent operators, then the equation for $\mathcal{Q}$ can be decomposed as a coupled system.
\begin{theorem}[Twisted Marchenko lemma]
\label{2c_marchenko_projection}
Assume that $P_1, P_2$ are orthogonal idempotents summing to the identity, i.e.
\begin{align}\label{opera_p}
P_1+P_2=1,\quad P_\ell^2=P_\ell,\quad P_1P_2=P_2P_1=0,
\end{align}
and the projections of $\mathcal{Q}=(\partial\Gamma)\Gamma^{-1}$ are defined by $\mathcal{G}_\ell=P_\ell\mathcal{Q}$ ($\ell=1,2$). Then $\mathcal{G}_\ell$ ($\ell=1,2$) satisfy
\begin{subequations}
\begin{align}
(\partial_1^{(\ell)})^2\mathcal G_\ell-\partial_2^{(\ell)}\mathcal G_\ell
&=-2\bigl(\partial_1^{(\ell)}\mathcal G_\ell\bigr)\mathcal G_\ell,
\label{marchenko_2}\\
(\partial_1^{(\ell)})^2\mathcal G_{\bar\ell}
+\partial_2^{(\ell)}\mathcal G_{\bar\ell}
&=-2\mathcal G_{\bar\ell}\partial_1^{(\ell)}\mathcal G_\ell,
\label{marchenko_3}\\
\partial_1^{(\ell)}\mathcal G_{\bar\ell}
&=-\mathcal G_{\bar\ell}\mathcal G_\ell,
\label{marchenko_1}
\end{align}
\end{subequations}
where $\ell+\bar{\ell}\equiv 0 \pmod{3}$.
\end{theorem}

\begin{proof}
Recall that $P_\ell$ ($\ell=1,2$) satisfy \eqref{opera_p} and $\mathcal{Q}=\mathcal{G}_1+\mathcal{G}_2$, then by multiplying~\eqref{ncq1} on the left by $P_\ell$ and by $P_{\bar\ell}$ respectively, we obtain
\begin{subequations}
\begin{align}
\partial_1^{(\ell)}\mathcal G_\ell
&=\partial\mathcal G_\ell
+\mathcal G_\ell\mathcal G_{\bar\ell},\label{first_same}\\
\partial_1^{(\ell)}\mathcal G_{\bar\ell}
&=-\mathcal G_{\bar\ell}\mathcal G_\ell,\label{first_cross}
\end{align}
\end{subequations}
where the second equation is precisely \eqref{marchenko_1}.
On the other hand, from~\eqref{ncq2} we get
\begin{align}
\label{Q_second_p}
\partial_2^{(\ell)}\mathcal G_\ell
=\partial^2\mathcal G_\ell
+2(\partial\mathcal G_\ell)(\mathcal G_\ell+\mathcal G_{\bar\ell})
+\mathcal G_\ell\left[\partial\mathcal G_{\bar\ell}+\mathcal G_{\bar\ell}(\mathcal G_\ell+\mathcal G_{\bar\ell})\right].
\end{align}
Applying the derivation $\partial_1^{(\ell)}$ to~\eqref{first_same} gives
\[
(\partial_1^{(\ell)})^2\mathcal G_\ell
=\partial\partial_1^{(\ell)}\mathcal G_\ell
+\bigl(\partial_1^{(\ell)}\mathcal G_\ell\bigr)\mathcal G_{\bar\ell}
+\mathcal G_\ell\partial_1^{(\ell)}\mathcal G_{\bar\ell}.
\]
Substituting \eqref{first_same} and \eqref{first_cross} into the right-hand side, we find
\begin{equation}
(\partial_1^{(\ell)})^2\mathcal G_\ell
=\partial^2\mathcal G_\ell
+2(\partial\mathcal G_\ell)\mathcal G_{\bar\ell}
+\mathcal G_\ell\bigl[\partial\mathcal G_{\bar\ell}
+\mathcal G_{\bar\ell}(\mathcal G_{\bar\ell}-\mathcal G_\ell)\bigr],
\label{first_same_d}
\end{equation}
where we use the fact that $\partial$ and $\partial_1^{(\ell)}$ commute.

Moreover, from~\eqref{first_same} we also have
\begin{align}
2\bigl(\partial_1^{(\ell)}\mathcal G_\ell\bigr)\mathcal G_\ell
=2(\partial\mathcal G_\ell)\mathcal G_\ell
+2\mathcal G_\ell\mathcal G_{\bar\ell}\mathcal G_\ell.
\label{first_same_product}
\end{align}
Adding \eqref{first_same_d} and \eqref{first_same_product}, we obtain
\[
(\partial_1^{(\ell)})^2\mathcal G_\ell
+2\bigl(\partial_1^{(\ell)}\mathcal G_\ell\bigr)\mathcal G_\ell
=\partial^2\mathcal G_\ell
+2(\partial\mathcal G_\ell)(\mathcal G_\ell+\mathcal G_{\bar\ell})
+\mathcal G_\ell\left[\partial\mathcal G_{\bar\ell}+\mathcal G_{\bar\ell}(\mathcal G_\ell+\mathcal G_{\bar\ell})\right].
\]
Comparing this with \eqref{Q_second_p}, we conclude that
\[
\partial_2^{(\ell)}\mathcal G_\ell
=(\partial_1^{(\ell)})^2\mathcal G_\ell
+2\bigl(\partial_1^{(\ell)}\mathcal G_\ell\bigr)\mathcal G_\ell,
\]
which is precisely \eqref{marchenko_2}.

Finally, multiplying \eqref{ncq2} on the left by $P_{\bar\ell}$ gives
\[
\partial_2^{(\ell)}\mathcal G_{\bar\ell}
=-\mathcal G_{\bar\ell}\left(\partial\mathcal G_\ell+\mathcal G_\ell(\mathcal G_\ell+\mathcal G_{\bar\ell})\right).
\]
Replacing $\partial\mathcal G_\ell$ by $\partial_1^{(\ell)}\mathcal G_\ell-\mathcal G_\ell\mathcal G_{\bar\ell}$, we obtain
\begin{align}
\label{second_cross_reduced}
\partial_2^{(\ell)}\mathcal G_{\bar\ell}
=-\mathcal G_{\bar\ell}\left(\partial_1^{(\ell)}\mathcal G_\ell+\mathcal G_\ell^2\right).
\end{align}
Applying the derivation $\partial_1^{(\ell)}$ to~\eqref{first_cross} yields
\begin{equation}
\label{first_cross_d}
(\partial_1^{(\ell)})^2\mathcal G_{\bar\ell}
=\mathcal G_{\bar\ell}\mathcal G_\ell^2
-\mathcal G_{\bar\ell}
(\partial_1^{(\ell)}\mathcal G_\ell).
\end{equation}
Adding~\eqref{second_cross_reduced} and~\eqref{first_cross_d}, we find
\[
(\partial_1^{(\ell)})^2\mathcal G_{\bar\ell}
+\partial_2^{(\ell)}\mathcal G_{\bar\ell}
=-2\mathcal G_{\bar\ell}\partial_1^{(\ell)}\mathcal G_\ell,
\]
which proves \eqref{marchenko_3}.
\end{proof}

The twisted Marchenko lemma admits two different types of reductions, which are the single-component reduction and the involution reduction. For the
single-component reduction, we may simply take $\mathcal G_2=0$, from which we have the following result.
\begin{coro}
If $\Gamma=W_n(\bphi)$ satisfies 
\begin{align}\label{gamma_burgers}
 \partial_1\Gamma-\partial\Gamma=\Gamma A,\quad \partial_2 \Gamma-\partial^2\Gamma=\Gamma C,
\end{align}
where $A$ and $C$ are constant matrices, then $\mathcal{G}=(\partial \Gamma)\Gamma^{-1}$ satisfies the non-commutative Burgers equation
\begin{align*}
\partial_1^2\mathcal{G}-\partial_2 \mathcal{G}=-2(\partial_1\mathcal{G})\mathcal{G}.
\end{align*}
\end{coro}
This result is consistent with \cite{hamanaka032} if we take $\partial_2=\partial_t$ and $\partial_1=\partial_x$. Moreover, it is also a generalization of the Example \ref{example}.

For the involution reduction, we can consider the difference between $\ell=1$ and $\ell=2$ in \eqref{marchenko_second} and \eqref{marchenko_first}, from which we get
\begin{align}\label{gamma_ds}
\partial_1\Gamma+B\partial\Gamma=\Gamma A,\quad \partial_2\Gamma+B\partial^2\Gamma=\Gamma C,
\end{align}
where 
$
\partial_k=\partial_k^{(2)}-\partial_k^{(1)}$ ($k=1,\,2$), $A=A_2-A_1$, $C=C_2-C_1
$,
and 
$
B=P_1-P_2.
$
We call this reduction as involution reduction since we have 
\begin{align*}
B^2=(P_1-P_2)^2=P_1+P_2=1.
\end{align*}
Different from the Wronski matrix given in \eqref{gamma_burgers}, the matrix in \eqref{gamma_ds} is still a two-component Wronski matrix $W_{n,m}(\bphi,\bpsi)$. In this way, the nonlinear variable $\mathcal{Q}=(\partial \Gamma)\Gamma^{-1}$ satisfies the system
\begin{subequations} \label{EGR_Q}
\begin{align}
&\partial_2\mathcal Q+B\left(\partial^2\mathcal Q+2(\partial\mathcal Q)\mathcal Q\right)+\bigl[B,\mathcal Q\bigr]\left(\partial\mathcal Q+\mathcal Q^2\right)=0,\label{EGR_Q2} \\
&\partial_1\mathcal Q+B\partial\mathcal Q+\bigl[B,\mathcal Q\bigr]\mathcal Q=0.\label{EGR_Q1}
\end{align}
\end{subequations}
Recall that $\mathcal{Q}=\mathcal{G}_1+\mathcal{G}_2$ and $B\mathcal{Q}=\mathcal{G}_1-\mathcal{G}_2$. Therefore, we can set the variable transformations
\begin{align*}
\mathcal F_1:=\mathcal Q+B\mathcal Q=2\mathcal G_1,
\qquad
\mathcal F_2:=\mathcal Q-B\mathcal Q=2\mathcal G_2,
\end{align*}
which satisfy the following proposition.
\begin{proposition} \label{involution_reduction}
The nonlinear variables $\mathcal F_1$ and $\mathcal F_2$ satisfy the following equations
\begin{subequations}
\begin{align}
&2\partial_2\mathcal F_1
+(\partial_1^2+\partial^2)\mathcal F_1
-2(\partial_1\mathcal F_1)\mathcal F_1
+2\mathcal F_1(\partial_1\mathcal F_2)
-2\mathcal F_1\mathcal F_2\mathcal F_1=0,\label{ds_f1}\\
&2\partial_2\mathcal F_2
-(\partial_1^2+\partial^2)\mathcal F_2
-2(\partial_1\mathcal F_2)\mathcal F_2
+2\mathcal F_2(\partial_1\mathcal F_1)
+2\mathcal F_2\mathcal F_1\mathcal F_2=0,\label{ds_f2}\\
&\partial_1\mathcal F_1+\partial\mathcal F_1
+\mathcal F_1\mathcal F_2=0,\label{ds_f3}\\
&\partial_1\mathcal F_2-\partial\mathcal F_2
-\mathcal F_2\mathcal F_1=0.\label{ds_f4}
\end{align}
\end{subequations}
\end{proposition}

\begin{proof}
From $B^2=1$ and the definitions $\mathcal F_1=(1+B)\mathcal Q$, $\mathcal F_2=(1-B)\mathcal Q$, we obtain
\begin{align*}
B[B,\mathcal Q]=-[B,\mathcal Q]B,\qquad
(1+B)[B,\mathcal Q]=\mathcal F_1(1-B),\qquad
(1-B)[B,\mathcal Q]=-\mathcal F_2(1+B).
\end{align*}

Taking the sum of \eqref{EGR_Q1} and its left multiplication by $B$ gives
\begin{align*}
0=(1+B)\left(\partial_1\mathcal Q+B\partial\mathcal Q+[B,\mathcal Q]\mathcal Q\right)
=\partial_1\mathcal F_1+\partial\mathcal F_1+\mathcal F_1\mathcal F_2.
\end{align*}
Similarly, taking the difference yields
\begin{align*}
0=(1-B)\left(\partial_1\mathcal Q+B\partial\mathcal Q+[B,\mathcal Q]\mathcal Q\right)
=\partial_1\mathcal F_2-\partial\mathcal F_2-\mathcal F_2\mathcal F_1,
\end{align*}
which correspond to \eqref{ds_f3} and \eqref{ds_f4}, respectively.

Applying $\partial_1$ to \eqref{EGR_Q1} and eliminating
$\partial_1\mathcal Q$ via the relation
$\partial_1\mathcal Q=-B\partial\mathcal Q-[B,\mathcal Q]\mathcal Q$ yields
\begin{equation} \label{Q_auxiliary}
(\partial_1^2-\partial^2)\mathcal Q-2\Bigl(B[B,\partial\mathcal Q]\mathcal Q
+B[B,\mathcal Q]\mathcal Q^2+[B,\mathcal Q]^2\mathcal Q\Bigr)=0.
\end{equation}
Adding twice \eqref{EGR_Q2} to $B$ times \eqref{Q_auxiliary} eliminates the term $[B,\mathcal Q]\mathcal Q^2$ and gives
\begin{equation} \label{Q_balanced}
2\partial_2\mathcal Q+B(\partial_1^2+\partial^2)\mathcal Q
+2\bigl(B\partial\mathcal Q+(\partial\mathcal Q)B\bigr)\mathcal Q
+2[B,\mathcal Q]\bigl(\partial\mathcal Q+B[B,\mathcal Q]\mathcal Q\bigr)=0.
\end{equation}
Left multiplication of \eqref{Q_balanced} by $B$, combined with \eqref{EGR_Q1} to eliminate $\partial\mathcal Q$, leads to
\begin{equation} \label{Q_reduced}
2B\partial_2\mathcal Q+(\partial_1^2+\partial^2)\mathcal Q
+2[B,\mathcal Q]^2\mathcal Q
-2\bigl(B\partial_1\mathcal Q+(\partial_1\mathcal Q)B\bigr)\mathcal Q
+2[B,\mathcal Q]\partial_1\mathcal Q=0.
\end{equation}
Finally, applying $1+B$ and $-(1-B)$ from the left to \eqref{Q_reduced}, respectively, we obtain
\begin{align*}
0={}&2\partial_2\mathcal F_1+(\partial_1^2+\partial^2)\mathcal F_1
-2(\partial_1\mathcal F_1)\mathcal F_1
+2\mathcal F_1\partial_1\mathcal F_2
-2\mathcal F_1\mathcal F_2\mathcal F_1,\\
0={}&2\partial_2\mathcal F_2-(\partial_1^2+\partial^2)\mathcal F_2
-2(\partial_1\mathcal F_2)\mathcal F_2
+2\mathcal F_2\partial_1\mathcal F_1
+2\mathcal F_2\mathcal F_1\mathcal F_2,
\end{align*}
which are precisely \eqref{ds_f1} and \eqref{ds_f2}, respectively.
\end{proof}

\begin{remark}
There is another choice of variable transformation. If we take $U=\mathcal{Q}B-B\mathcal{Q}$ and $V=\mathcal{Q}B+B\mathcal{Q}$, then $U$ and $V$ satisfy
\begin{align*}
&2B\partial_2U+(\partial_1^2+\partial^2)U+2U^3-2(U\partial_1 V+(\partial_1V)U)=0, \\
&\partial_1V +B\partial V=U^2.
\end{align*}
This system coincides with~\cite[Equations (19) \& (20)]{etingof98}.    
\end{remark}

\subsection{Specialization to the non-commutative two-component KP equation}\label{sec:3.2}
We now apply Theorem~\ref{2c_marchenko_projection} to the non-commutative two-component KP equation. Let
$W_{n,m}(\boldsymbol\phi,\boldsymbol\psi)$ be the non-degenerate two-component Wronski matrix defined in \eqref{2c_wronski}, and set $\Gamma=W_{n,m}(\bphi,\bpsi)$. We assume that $\{\phi_j\}_{j=1}^{n+m}$ are independent with $\partial^{(2)}$, while $\{\psi_j\}_{j=1}^{n+m}$ are independent with $\partial^{(1)}$. More precisely, we impose the derivative relations\footnote{Without abuse of notation, we henceforth write $\partial_{t_j^{(k)}}$ and $\partial_x$ in place of $\partial_j^{(k)}$ and $\partial$ respectively, in order to distinguish the time flows and spectral parameter when presenting the corresponding integrable equations.}
\begin{align}
\label{seed-dispersion}
\partial_{t_k^{(1)}}\phi_j=\partial_x^k\phi_j,
\quad
\partial_{t_k^{(2)}}\phi_j=0, \quad
\partial_{t_k^{(1)}}\psi_j=0,
\quad
\partial_{t_k^{(2)}}\psi_j=\partial_x^k\psi_j,
\end{align}
for $k=1,2$. These derivative relations can be reformulated as the Marchenko lemma conditions \eqref{marchenko_second} and \eqref{marchenko_first} by recognizing 
\begin{align}\label{wpp}
P_1=\operatorname{diag}(I_n,0_m), \quad 
P_2=\operatorname{diag}(0_n,I_m),
\end{align}
and $A_\ell=C_\ell=0$ for $\ell=1,2$,
where $I_k$ and $0_k$ denote the identity matrix and the zero matrix of order $k$, respectively. 
\begin{remark}
Indeed, there should be a large class of dispersion relations. From the linear relations \eqref{marchenko_second} and \eqref{marchenko_first}, one knows that for constant matrix $P_1=\operatorname{diag}(I_n,0_m), \,
P_2=\operatorname{diag}(0_n,I_m)$, $A=(a_{ij})$ and $C=(c_{ij})$, $\phi_j$ and $\psi_j$  should satisfy 
\begin{align}\label{time_KP}
\begin{aligned}
&\partial_{t_1^{(1)}}\phi_j=\partial_x\phi_j+\sum_{k=1}^{n+m} \phi_k a_{k,j},&&\partial_{t_1^{(2)}}\phi_j=\sum_{k=1}^{n+m} \phi_k a_{k,j},\\
&\partial_{t_2^{(1)}}\phi_j=\partial_x^2\phi_j+\sum_{k=1}^{n+m} \phi_k c_{k,j},&& \partial_{t_2^{(2)}}\phi_j=\sum_{k=1}^{n+m} \phi_k c_{k,j},\\
&\partial_{t_1^{(1)}}\psi_j=\sum_{k=1}^{n+m}\psi_ka_{k,j},&&\partial_{t_1^{(2)}}\psi_j=\partial_x\psi_j+\sum_{k=1}^{n+m}\psi_k a_{k,j},\\
&\partial_{t_2^{(1)}}\psi_j=\sum_{k=1}^{n+m}\psi_k c_{k,j},&&\partial_{t_2^{(2)}}\psi_j=\partial_x^2\psi_j+\sum_{k=1}^{n+m}\psi_k c_{k,j}.
\end{aligned}
\end{align}
\end{remark} 

Recall that $\xi_j$ and $\eta_j$, $j=1,\ldots,n+m$, are the coefficients of the coupled system \eqref{coupled_system}. Then Corollary~\ref{wronski_logarithmic_derivative} and Theorem~\ref{2c_marchenko_projection} imply the following result.
\begin{proposition}
Let 
\begin{equation}
\label{four_fields}
q=-\xi_{n+m},
\qquad
\tilde q=-\xi_n,
\qquad
r=-\eta_n,
\qquad
\tilde r=-\eta_{n+m},
\end{equation}
then these four functions satisfy the non-commutative two-component KP equation
\begin{align}
\label{2cKP}
\begin{aligned}
\partial_{t_1^{(1)}}^2q-\partial_{t_2^{(1)}}q&=
2(\partial_{t_1^{(1)}}\tilde q)q, \\
\partial_{t_1^{(2)}}^2q+\partial_{t_2^{(2)}}q&=
2q\,\partial_{t_1^{(2)}}\tilde r, \\
\partial_{t_1^{(1)}}^2r+\partial_{t_2^{(1)}}r&=
2r\,\partial_{t_1^{(1)}}\tilde q, \\
\partial_{t_1^{(2)}}^2r-\partial_{t_2^{(2)}}r&=
2(\partial_{t_1^{(2)}}\tilde r)r, \\
\partial_{t_1^{(2)}}\tilde q&=qr, \\  
\partial_{t_1^{(1)}}\tilde r&=rq.
\end{aligned}
\end{align}
\end{proposition}
\begin{proof}
By Corollary~\ref{wronski_logarithmic_derivative}, the $n$-th
and $(n+m)$-th rows of $\mathcal Q=(\partial W)W^{-1}$ are $(\xi_1,\ldots,\xi_{n+m})$ and $(\eta_1,\ldots,\eta_{n+m})$, respectively. Since
$\mathcal G_\ell=P_\ell\mathcal Q$, we have
\[
(\mathcal G_1)_{n,n}=\xi_n,\qquad
(\mathcal G_1)_{n,n+m}=\xi_{n+m},
\]
and
\[
(\mathcal G_2)_{n+m,n}=\eta_n,\qquad
(\mathcal G_2)_{n+m,n+m}=\eta_{n+m}.
\]
Taking the $(n+m,n+m)$-entries of \eqref{marchenko_1} for
$\ell=1$ and the $(n,n)$-entries for $\ell=2$ yields
\begin{align*}
\partial_1^{(1)}\eta_{n+m}
=-\eta_n\,\xi_{n+m}, \qquad
\partial_1^{(2)}\xi_n
=-\xi_{n+m}\,\eta_n.
\end{align*}
Similarly, taking the
$(n,n+m)$-entries of \eqref{marchenko_2} for $\ell=1$ and
\eqref{marchenko_3} for $\ell=2$, and the $(n+m,n)$-entries of
\eqref{marchenko_3} for $\ell=1$ and
\eqref{marchenko_2} for $\ell=2$, we obtain
\begin{align*}
&\partial_{t_1^{(1)}}^2\xi_{n+m}
-\partial_{t_2^{(1)}}\xi_{n+m}
=-2(\partial_{t_1^{(1)}}\xi_n)\,\xi_{n+m},
\\
&\partial_{t_1^{(2)}}^2\xi_{n+m}
+\partial_{t_2^{(2)}}\xi_{n+m}
=-2\xi_{n+m}\partial_{t_1^{(2)}}\eta_{n+m},
\\
&\partial_{t_1^{(1)}}^2\eta_n
+\partial_{t_2^{(1)}}\eta_n
=-2\eta_n\,\partial_{t_1^{(1)}}\xi_n,
\\
&\partial_{t_1^{(2)}}^2\eta_n
-\partial_{t_2^{(2)}}\eta_n
=-2(\partial_{t_1^{(2)}}\eta_{n+m})\,\eta_n.
\end{align*}
These equations are exactly the non-commutative two-component KP equation \eqref{2cKP} by relations \eqref{four_fields}.
\end{proof}
It is worth noting that, although the system \eqref{2cKP} involves four fields $(q,\tilde{q},r,\tilde{r})$, the fields $\tilde{q}$ and $\tilde{r}$ are not independent. From the differential relations
\begin{align*}
\partial_{t_1^{(2)}}\tilde{q}=qr,\quad \partial_{t_1^{(1)}}\tilde{r}=rq,
\end{align*}
we obtain, up to integration constants, that
\begin{align*}
\tilde{q}=\int qr dt_{1}^{(2)},\quad \tilde{r}=\int rq dt_1^{(1)}.
\end{align*}
Therefore, the essential dynamical degrees of freedom are captured by $(q,r)$. Explicit solutions for all four fields in terms of quasi-determinants have already been presented as follows.
\begin{coro}
The non-commutative two-component KP equation \eqref{2cKP} has the solution 
\begin{align*}
q=\left|\begin{array}{cc}
W_{n,m}(\bphi,\bpsi)&e_{n+m}\\
\bphi^{(n)}&\boxed{0}
\end{array}
\right|,\quad  r=\left|\begin{array}{cc}
W_{n,m}(\bphi,\bpsi)&e_{n}\\
\bpsi^{(m)}&\boxed{0}
\end{array}
\right|
\end{align*}
with general dispersion relations \eqref{time_KP}.
\end{coro}

We refer to the system~\eqref{2cKP} as the non-commutative version of the two-component KP equation. In the commutative limit, the quasi-determinant representations of $q,\tilde q,r$, and $\tilde r$ reduce to
\begin{align*}
&q = (-1)^m \frac{\tau_{n+1,m-1}}{\tau_{n,m}},  \qquad \quad \tilde{q} = -\frac{\partial_{t_1^{(1)}}\tau_{n,m}}{\tau_{n,m}}, \\
&r = (-1)^{m+1} \frac{\tau_{n-1,m+1}}{\tau_{n,m}}, \qquad \tilde{r} = - \frac{\partial_{t_1^{(2)}}\tau_{n,m}}{\tau_{n,m}}.     
\end{align*}
where $\tau_{n,m}=\det W_{n,m}(\boldsymbol\phi,\boldsymbol\psi)$. 
Consequently, the commutative limit of \eqref{2cKP} takes the form
\begin{equation*} 
\begin{split}
&(\mathrm{D}_{t_1^{(1)}}^2 + \mathrm{D}_{t_2^{(1)}}) \tau_{n,m} \cdot \tau_{n+1,m-1} = 0, \\
&(\mathrm{D}_{t_1^{(2)}}^2 - \mathrm{D}_{t_2^{(2)}}) \tau_{n,m} \cdot \tau_{n+1,m-1} = 0, \\
&(\mathrm{D}_{t_1^{(1)}}^2 - \mathrm{D}_{t_2^{(1)}}) \tau_{n,m} \cdot \tau_{n-1,m+1} = 0, \\
&(\mathrm{D}_{t_1^{(2)}}^2 + \mathrm{D}_{t_2^{(2)}}) \tau_{n,m} \cdot \tau_{n-1,m+1} = 0, \\
&\mathrm{D}_{t_1^{(1)}} \mathrm{D}_{t_1^{(2)}} \tau_{n,m} \cdot \tau_{n,m} = 2 \tau_{n+1,m-1} \tau_{n-1,m+1} ,
\end{split}
\end{equation*}
where $\mathrm D$ denotes the corresponding Hirota bilinear operator, defined by
\begin{align*}
\mathrm D_x^m \mathrm D_y^n f\cdot g
=(\partial_x-\partial_{x'})^m(\partial_y-\partial_{y'})^n
f(x,y)\,g(x',y')\big|_{x'=x,\ y'=y}.
\end{align*}
The above system is precisely the bilinear form of the two-component KP
equation, as given in Refs.~\cite{jimbo83,freeman90}.

\subsection{Specialization to the non-commutative DS equation}\label{sec:3.3}

We now apply the involution reduction established in Proposition~\ref{involution_reduction} to derive the non-commutative DS equation. In this subsection we set
\begin{equation*}
\partial_1=\sigma\partial_y,\qquad
\partial_2=\frac{\mathrm i}{2}\partial_t,
\end{equation*}
where $\sigma=-1\ \text{or}\ \sigma=\mathrm i$, so that the difference derivations of the involution reduction correspond to one spatial and one temporal flow.

For each $j=1,\ldots,n+m$, let $\phi_j$ and $\psi_j$ satisfy the dispersion relations
\begin{align} \label{seed-dispersion-ds}
\sigma\partial_y\phi_j=-\partial_x\phi_j,\qquad
\partial_t\phi_j=2\mathrm i\partial_x^2\phi_j,\qquad
\sigma\partial_y\psi_j=\partial_x\psi_j,\qquad
\partial_t\psi_j=-2\mathrm i\partial_x^2\psi_j,
\end{align}
and consider the two-component Wronski matrix
$W_{n,m}(\boldsymbol\phi,\boldsymbol\psi)$ defined in \eqref{2c_wronski}.
Assuming that $W_{n,m}(\boldsymbol\phi,\boldsymbol\psi)$ is invertible, we set
\begin{align*}
\Gamma=W_{n,m}(\boldsymbol\phi,\boldsymbol\psi),
\qquad
B=\operatorname{diag}(I_n,-I_m).
\end{align*}
Since $B$ acts as the identity on the first $n$ rows of
$W_{n,m}(\boldsymbol\phi,\boldsymbol\psi)$ (the $\phi$-rows) and as $-I_m$ on
the remaining $m$ rows (the $\psi$-rows), the dispersion relations 
\eqref{seed-dispersion-ds} imply
\begin{align*}
\sigma\partial_y W_{n,m}(\boldsymbol\phi,\boldsymbol\psi)
=-B\partial_x W_{n,m}(\boldsymbol\phi,\boldsymbol\psi),\qquad
\frac{\mathrm i}{2}\partial_t W_{n,m}(\boldsymbol\phi,\boldsymbol\psi)
=-B\partial_x^2 W_{n,m}(\boldsymbol\phi,\boldsymbol\psi),
\end{align*}
which are precisely the Marchenko conditions with $A=C=0$. Consequently, Proposition~\ref{involution_reduction} applies. With $\mathcal Q=(\partial_x W_{n,m}(\boldsymbol\phi,\boldsymbol\psi))
W_{n,m}(\boldsymbol\phi,\boldsymbol\psi)^{-1}$, the matrices
$\mathcal F_1=\mathcal Q+B\mathcal Q$ and $\mathcal F_2=\mathcal Q-B\mathcal Q$
satisfy \eqref{ds_f1}-\eqref{ds_f4}.

By Corollary~\ref{wronski_logarithmic_derivative}, the $n$-th and
$(n+m)$-th rows of $\mathcal Q$ are given by the vectors
$(\xi_1,\ldots,\xi_{n+m})$ and $(\eta_1,\ldots,\eta_{n+m})$, respectively, where
$\xi_j,\eta_j$ are the coefficients in \eqref{coupled_system} determined by the
chosen seed functions through Theorem~\ref{coefficients}. Analogously to \eqref{four_fields}, we define
\begin{equation}\label{ds_fields}
q=-\xi_{n+m},\qquad
\tilde q=-\xi_n,\qquad
r=-\eta_n,\qquad
\tilde r=-\eta_{n+m}.
\end{equation}
Since $\mathcal F_1=\mathcal Q+B\mathcal Q$ and
$\mathcal F_2=\mathcal Q-B\mathcal Q$, the first $n$ rows of
$\mathcal Q$ are doubled in $\mathcal F_1$ and the last $m$ rows of
$\mathcal Q$ are doubled in $\mathcal F_2$. Consequently, the fields
$q,\tilde q,r,\tilde r$ defined in \eqref{ds_fields} appear in the
entries of $\mathcal F_1$ and $\mathcal F_2$ as
\begin{align*}
(\mathcal F_1)_{n,n}=-2\tilde q,\qquad
(\mathcal F_1)_{n,n+m}=-2q,\qquad
(\mathcal F_2)_{n+m,n}=-2r,\qquad
(\mathcal F_2)_{n+m,n+m}=-2\tilde r.
\end{align*}

\begin{proposition}
The fields $q,\tilde q,r,\tilde r$ defined in \eqref{ds_fields} satisfy
\begin{align}
\label{ncds_marchenko}
\begin{aligned}
&\mathrm i\partial_t q+(\partial_x^2+\sigma^2\partial_y^2)q
-4\sigma \bigl(q\partial_y\tilde r-(\partial_y\tilde q)q\bigr)-8qrq=0, \\
&\mathrm i\partial_t r-(\partial_x^2+\sigma^2\partial_y^2)r
-4\sigma \bigl(r\partial_y\tilde q-(\partial_y\tilde r)r\bigr)+8rqr=0, \\
&(\partial_x+\sigma\partial_y)\tilde q-2qr=0,\\
&(\partial_x-\sigma\partial_y)\tilde r-2rq=0.
\end{aligned}
\end{align}
\end{proposition}

\begin{proof}
By Corollary~\ref{wronski_logarithmic_derivative}, the only rows of
$\mathcal Q$ that are not shift rows are the $n$-th and $(n+m)$-th rows, which
contain the coefficients $\xi_j$ and $\eta_j$. Consequently, every
product below depends only on the fields defined in \eqref{ds_fields}. A direct computation gives
\begin{align*}
&(\mathcal F_1\mathcal F_2)_{n,n}=4qr,
&&(\mathcal F_2\mathcal F_1)_{n+m,n+m}=4rq,\\
&(\mathcal F_1\mathcal F_2\mathcal F_1)_{n,n+m}=-8qrq,
&&(\mathcal F_2\mathcal F_1\mathcal F_2)_{n+m,n}=-8rqr,\\
&\bigl((\partial_1\mathcal F_1)\mathcal F_1\bigr)_{n,n+m}
=4(\partial_1\tilde q)q,
&&\bigl(\mathcal F_1\partial_1\mathcal F_2\bigr)_{n,n+m}
=4q\,\partial_1\tilde r,\\
&\bigl((\partial_1\mathcal F_2)\mathcal F_2\bigr)_{n+m,n}
=4(\partial_1\tilde r)r,
&&\bigl(\mathcal F_2\partial_1\mathcal F_1\bigr)_{n+m,n}
=4r\,\partial_1\tilde q.
\end{align*}
Taking the $(n,n)$- and $(n+m,n+m)$-entries of \eqref{ds_f3} and \eqref{ds_f4}
, respectively, yields the last two equations in \eqref{ncds_marchenko}. Similarly, extracting the $(n,n+m)$-entry of \eqref{ds_f1} and dividing by $-2$,
together with the $(n+m,n)$-entry of \eqref{ds_f2} multiplied by $-\tfrac12$,
gives the first two equations of \eqref{ncds_marchenko}.
\end{proof}

\begin{coro}
If the seed functions satisfy 
\begin{align*}
&(\sigma\partial_y+\partial_x)\phi_j=\sum_{k=1}^{n+m} \phi_k a_{k,j},
\quad \left(
\frac{i}{2}\partial_t+\partial_x^2
\right)\phi_j=\sum_{k=1}^{n+m}\phi_k c_{k,j},\\
&(\sigma\partial_y-\partial_x)\psi_j=\sum_{k=1}^{n+m} \psi_k a_{k,j},
\quad \left(
\frac{i}{2}\partial_t-\partial_x^2
\right)\psi_j=\sum_{k=1}^{n+m}\psi_k c_{k,j},
\end{align*}
then
\eqref{ncds_marchenko} admits the quasi-Wronskian solution
\begin{align*}
q&=
\left|\begin{array}{cc}
W_{n,m}(\bphi,\bpsi)&\boldsymbol e_{n+m}\\
\bphi^{(n)}&\boxed{0}
\end{array}\right|,
\qquad
r=
\left|\begin{array}{cc}
W_{n,m}(\bphi,\bpsi)&\boldsymbol e_n\\
\bpsi^{(m)}&\boxed{0}
\end{array}\right|.
\end{align*}
\end{coro}

\section{A direct approach to the non-commutative two-component KP hierarchy}\label{sec:4}

In the preceding section, the non-commutative two-component KP equation was obtained by applying the twisted Marchenko lemma to the Wronski matrix $W_{n,m}(\bphi,\bpsi)$. 
In this section, we adopt the direct method developed in~\cite{li26} to verify the solutions of the non-commutative two-component KP equation directly, and to demonstrate that two-component quasi-Wronskians provide a systematic tool for constructing higher flows.

\subsection{Two-component quasi-Wronskians}
For $\ell=1, 2$, let 
\[
\boldsymbol{t}^{(\ell)}=(t_1^{(\ell)}, t_2^{(\ell)}, \ldots)
\]
be two sequences of independent variables. For $i=1,\ldots,n+m$, let 
\[
\phi_i=\phi_i(x,\boldsymbol{t}^{(1)}), \qquad  
\psi_i=\psi_i(x,\boldsymbol{t}^{(2)})
\] 
be functions taking values in a (not necessarily commutative) unital associative algebra $\mathcal{A}$, satisfying the dispersion relations
\begin{equation}
\partial_{t_p^{(1)}}\phi_j
=\partial_x^p\phi_j,
\quad
\partial_{t_p^{(2)}}\phi_j
=0,\quad 
\partial_{t_p^{(1)}}\psi_j
=0,
\quad
\partial_{t_p^{(2)}}\psi_j
=\partial_x^p\psi_j
\label{dispersion}
\end{equation}
for all $p\in\mathbb{N}_+$. These relations include the higher flows compared to the dispersion relations in \eqref{seed-dispersion}. A typical choice of seed functions is, for instance,
\begin{equation}
\begin{split}
\phi_j &= a_j \exp\left(k_j (t_1^{(1)}+x)+\sum_{p=2}^{\infty} k_j^p t_p^{(1)}\right)+b_j \exp\left(l_j (t_1^{(1)}+x)+\sum_{p=2}^{\infty} l_j^p t_p^{(1)}\right), \\
\psi_j &= c_j \exp\left(\alpha_j (t_1^{(2)}+x)+\sum_{p=2}^{\infty} \alpha_j^p t_p^{(2)}\right)+d_j \exp\left(\beta_j (t_1^{(2)}+x)+\sum_{p=2}^{\infty} \beta_j^p t_p^{(2)}\right),
\end{split} 
\end{equation}
where $a_j,b_j,c_j,d_j\in\mathcal{A}$ and $k_j,l_j,\alpha_j,\beta_j\in\mathcal{A}$.
With this notation, we introduce the following definition for generalized quasi-Wronskians.
\begin{definition} \label{quasi-Wronskians}
For $j=0,\cdots,m-1$ and $i\in\mathbb{N}$, we define the following generalized quasi-Wronskians
\begin{equation*}
Q^{i,j}
=
\begin{vmatrix}
W_{n,m}(\bphi,\bpsi) & e_{n+m-j} \\
\partial^{n+i}\boldsymbol{\phi}& \boxed{0}
\end{vmatrix},
\quad
\widetilde{R}^{i,j}
=
\begin{vmatrix}
W_{n,m}(\bphi,\bpsi) & e_{n+m-j} \\
\partial^{m+i}\boldsymbol{\psi} & \boxed{0}
\end{vmatrix},
\end{equation*}
and for $k=0,\cdots,n-1$ and $i\in\mathbb{N}$, we define
\begin{equation*}
\widetilde{Q}^{i,k}
=
\begin{vmatrix}
W_{n,m}(\bphi,\bpsi) & e_{n-k} \\
\partial^{n+i}\boldsymbol{\phi} & \boxed{0}
\end{vmatrix},
\quad
R^{i,k}
=
\begin{vmatrix}
W_{n,m}(\bphi,\bpsi) & e_{n-k} \\
\partial^{m+i}\boldsymbol{\psi} & \boxed{0}
\end{vmatrix}.
\end{equation*}
\end{definition}

From the dispersion relation of seed functions \eqref{dispersion}, we obtain the following derivative formulas for the two-component quasi-Wronskians.
\begin{proposition} \label{prop1}
If the seed functions $\phi_1, \phi_2, \ldots, \phi_{n+m}$ and $\psi_1, \psi_2, \ldots, \psi_{n+m}$ satisfy the dispersion relations \eqref{dispersion}, then for any $t_p^{(1)}-$flows and $t_p^{(2)}$-flows $(p\in\mathbb{N})$, the quasi-Wronskians $Q^{i,j}$, $\tilde{R}^{i,j}$, $\tilde{Q}^{i,k}$ and $R^{i,k}$ obey the following derivative rules
\begin{align}
\begin{aligned}
&\partial_{t_p^{(1)}} Q^{i, j} = Q^{i + p, j} + \sum_{a=0}^{p-1} \tilde{Q}^{i,a} Q^{p-a-1,j}, \quad
\partial_{t_p^{(2)}} Q^{i, j} = -Q^{i, j+p} + \sum_{a=0}^{p-1} Q^{i,a} \tilde{R}^{p-a-1,j}, \\
&\partial_{t_p^{(1)}} R^{i, k} = -R^{i, k+p} + \sum_{a=0}^{p-1} R^{i,a} \tilde{Q}^{p-a-1,k}, \quad
\partial_{t_p^{(2)}} R^{i, k} = R^{i + p, k} + \sum_{a=0}^{p-1} \tilde{R}^{i,a} R^{p-a-1,k},\\
&\partial_{t_p^{(1)}} \tilde{Q}^{i, k} = \tilde{Q}^{i+p,k} - \tilde{Q}^{i,k+p} + \sum_{a=0}^{p-1} \tilde{Q}^{i,a} \tilde{Q}^{p-a-1,k}, \quad
\partial_{t_p^{(2)}} \tilde{Q}^{i, k} = \sum_{a=0}^{p-1} Q^{i,a} R^{p-a-1,k},\\
&\partial_{t_p^{(1)}} \tilde{R}^{i, j} = \sum_{a=0}^{p-1} R^{i,a} Q^{p-a-1,j}, \quad
\partial_{t_p^{(2)}} \tilde{R}^{i, j} = \tilde{R}^{i+p,j} - \tilde{R}^{i,j+p} + \sum_{a=0}^{p-1} \tilde{R}^{i,a} \tilde{R}^{p-a-1,j}.
\end{aligned}
\end{align}
\end{proposition}
\begin{proof}
We prove the two identities for $Q^{i,j}$, and the rest can be verified similarly. Using the linear relations (\ref{dispersion}) and the quasi-determinant derivative formula (\ref{dquasiwronskian}), we obtain the derivative of $Q^{i,j}$ with respect to $t_p^{(1)}$ 
\begin{align*}
\partial_{t_p^{(1)}} Q^{i, j} 
&= \left| \begin{array}{cc}
W_{n,m}(\bphi,\bpsi) & e_{n+m-j} \\
\partial^{n+i+p}\boldsymbol{\phi} & \boxed{0}
\end{array} \right| 
+ \sum_{a=1}^{n} 
\left| \begin{array}{cc}
W_{n,m}(\bphi,\bpsi) & e_a \\
\partial^{n+i}\boldsymbol{\phi} & \boxed{0}
\end{array} \right|
\left| \begin{array}{cc}
W_{n,m}(\bphi,\bpsi) & e_{n+m-j} \\
\partial^{a-1+p}\boldsymbol{\phi} & \boxed{0}
\end{array} \right| \\
&= Q^{i + p, j} + \sum_{a=0}^{n-1} \tilde{Q}^{i,a} Q^{p-1-a,j}.
\end{align*}
For $a>p-1$, it is easily checked that $Q^{p-1-a,j}=0$, then 
\begin{equation}
\partial_{t_p^{(1)}} Q^{i, j} = Q^{i + p, j} + \sum_{a=0}^{p-1} \tilde{Q}^{i,a} Q^{p-a-1,j}.
\end{equation}
Similarly, the derivative of $Q^{i,j}$ with respect to $t_p^{(2)}$ is given by
\begin{align*}
\partial_{t_p^{(2)}} Q^{i, j} 
&= \sum_{a=1}^{m} 
\left| \begin{array}{cc}
W_{n,m}(\bphi,\bpsi) & e_{n+a} \\
\partial^{n+i}\boldsymbol{\phi}& \boxed{0}
\end{array} \right|
\left| \begin{array}{cc}
W_{n,m}(\bphi,\bpsi) & e_{n+m-j} \\
\partial^{a+1-p}\boldsymbol{\psi} & \boxed{0}
\end{array} \right| \\
&= \sum_{a=0}^{m-1} Q^{i,a} \tilde{R}^{p-1-a,j}.
\end{align*}
For $a>p-1$ and $p-a+j\neq 0$, we have $\tilde{R}^{p-1-a,j}=0$, and for $p-a+j = 0$, we have $\tilde{R}^{p-1-a,j}=-1$. Consequently,
\begin{equation}
\partial_{t_p^{(2)}} Q^{i, j} = -Q^{i, j+p} + \sum_{a=0}^{p-1} Q^{i,a} \tilde{R}^{p-a-1,j}.
\end{equation}  
The derivative formulas for the remaining two-component quasi-Wronskians $\tilde{R}^{i,j}$, $\tilde{Q}^{i,k}$ and $R^{i,k}$ are obtained in analogous manners.
\end{proof}

\subsection{The construction of non-commutative two-component KP from quasi-Wronskians}
In the commutative case, the construction of the two-component KP equation via the $\tau$-function ultimately yields determinant identities and the Plücker relations. In contrast, within the non-commutative framework, in analogy with the one-component case, employing the quasi-Wronskians $Q^{0, 0}$, $\tilde{R}^{0, 0}$, $\tilde{Q}^{0, 0}$ and $R^{0, 0}$ to derive the non-commutative two-component KP equation leads to a complete cancellation of all $Q, \tilde{R}, \tilde{Q}, R$-terms except for $Q^{0,0}$, $\tilde{R}^{0,0}$, $\tilde{Q}^{0,0}$ and $R^{0,0}$. To facilitate this procedure, we begin by computing the derivatives of these quasi-Wronskians. For notational convenience, we set
\begin{equation}
\begin{split}
q=Q^{0,0}, \quad \tilde{r}=\tilde{R}^{0,0}, \quad \tilde{q}=\tilde{Q}^{0,0}, \quad r=R^{0,0}.       
\end{split}
\end{equation} 

From Proposition \ref{prop1}, the first-order derivatives of $q$, $\tilde{r}$, $\tilde{q}$ and $r$ with respect to $t_1^{(1)}$ and $t_1^{(2)}$ are
\begin{align}
\begin{aligned}
&\partial_{t_1^{(1)}} q = Q^{1,0} + \tilde{q} q, &&\partial_{t_1^{(2)}} q = - Q^{0,1} + q \tilde{r}, \\
&\partial_{t_1^{(1)}} \tilde{q} = \tilde{Q}^{1,0} - \tilde{Q}^{0,1} + \tilde{q}^2, &&
\partial_{t_1^{(2)}} \tilde{q} = q r, \label{first_2}\\
&\partial_{t_1^{(1)}} \tilde{r} = r q, &&
\partial_{t_1^{(2)}} \tilde{r} = \tilde{R}^{1,0} - \tilde{R}^{0,1} + \tilde{r}^2, \\
&\partial_{t_1^{(1)}} r = - R^{0,1} + r \tilde{q},&&
\partial_{t_1^{(2)}} r = R^{1,0} + \tilde{r} r. 
\end{aligned}
\end{align}
We  now seek suitable linear combinations of the derivative identities such that all $Q^{i,j}$, $\tilde{Q}^{i,j}$, $R^{i,j}$ and $\tilde{R}^{i,j}$ with $(i,j)\neq (0,0)$ are eliminated. 
Since the first-order identities alone do not yield the required closed
relations, we compute the second-order derivatives, which are listed below
\begin{align}
\begin{aligned}
&\partial_{t_1^{(1)}}^2 q = Q^{2,0} + (2\tilde{Q}^{1,0} - \tilde{Q}^{0,1}) q + \tilde{q} Q^{1,0} + 2 \tilde{q}^2 q, 
&&\partial_{t_1^{(2)}}^2 q = Q^{0,2} + q (\tilde{R}^{1,0} - 2\tilde{R}^{0,1}) - Q^{0,1} \tilde{r} + 2q \tilde{r}^2, \\
&\partial_{t_2^{(1)}} q = Q^{2,0} + \tilde{q} Q^{1,0} + \tilde{Q}^{0,1} q, 
&&\partial_{t_2^{(2)}} q = - Q^{0,2} + q \tilde{R}^{1,0} + Q^{0,1} \tilde{r}, \\
&\partial_{t_1^{(1)}}^2 r = R^{0,2} + r (\tilde{Q}^{1,0} - 2\tilde{Q}^{0,1}) - R^{0,1} \tilde{q} + 2r \tilde{q}^2, 
&&\partial_{t_1^{(2)}}^2 r = R^{2,0} + (2\tilde{R}^{1,0} - \tilde{R}^{0,1})r + \tilde{r} R^{1,0} + 2 \tilde{r}^2 r, \\
&\partial_{t_2^{(1)}} r = - R^{0,2} + r \tilde{Q}^{1,0} + R^{0,1} \tilde{q}, 
&&\partial_{t_2^{(2)}} r = R^{2,0} + \tilde{r} R^{1,0} + \tilde{R}^{0,1} r. 
\end{aligned}
\end{align}
Starting with $q$, we examine its highest-order derivatives, namely $\partial_{t_1^{(1)}}^2 q$,  $\partial_{t_1^{(2)}}^2 q$, $\partial_{t_2^{(1)}} q$ and $\partial_{t_2^{(2)}} q$. Among the terms on the right-hand side, we focus on those containing the monomials $Q^{i,j}$ up to $i+j=2$, in particular $Q^{2,0}$ and $Q^{0,2}$. These two monomials do not appear simultaneously in any single derivative equation. To eliminate them, we therefore consider the following linear combinations
\begin{align}
\begin{aligned}
&\partial_{t_1^{(1)}}^2 q - \partial_{t_2^{(1)}} q = 2 \left(\tilde{Q}^{1,0} - \tilde{Q}^{0,1} + \tilde{q}^2\right) q = 2 (\partial_{t_1^{(1)}} \tilde{q}) q, \\
&\partial_{t_1^{(2)}}^2 q + \partial_{t_2^{(2)}} q = 2 q \left(\tilde{R}^{1,0} - \tilde{R}^{0,1} + \tilde{r}^2\right) = 2 q (\partial_{t_1^{(2)}} \tilde{r}).
\end{aligned}
\end{align}
Similarly, from $\partial_{t_1^{(1)}}^2 r$, $\partial_{t_2^{(1)}} r$, $\partial_{t_1^{(2)}}^2 r$ and $\partial_{t_2^{(2)}} r$, we obtain
\begin{align}
\begin{aligned}
&\partial_{t_1^{(1)}}^2 r + \partial_{t_2^{(1)}} r = 2 r \left(\tilde{Q}^{1,0} - \tilde{Q}^{0,1} + \tilde{q}^2\right) = 2 r (\partial_{t_1^{(1)}} \tilde{q}), \\
&\partial_{t_1^{(2)}}^2 r - \partial_{t_2^{(2)}} r = 2 \left(\tilde{R}^{1,0} - \tilde{R}^{0,1} + \tilde{r}^2\right) r = 2 (\partial_{t_1^{(2)}} \tilde{r}) r.
\end{aligned}
\end{align}
Combining these results yields the closed system of nonlinear partial differential equations \eqref{2cKP}.

A broader class of higher-order equations in the non-commutative two-component KP hierarchy can be derived by the same method. In the following subsection, we present a higher-order non-commutative two-component KP equation obtained via this direct method.

\subsection{A higher-order non-commutative two-component KP equation}

Applying the same elimination procedure to the third-order derivatives listed in
Appendix~\ref{appendix:B}, we obtain the following higher-order closed system of
non-commutative two-component KP equations.

\begin{proposition}
\label{third_order_2cKP_theorem}
The two-component quasi-Wronskians $q,\tilde q,r,\tilde r$ satisfy the following
closed system of third-order non-commutative two-component KP equations
\begin{equation} \label{third_order_2cKP}
\begin{split}
&4\partial_{t_3^{(1)}}q
-\partial_{t_1^{(1)}}^3q
-3\partial_{t_1^{(1)}t_2^{(1)}}^2q
+6(\partial_{t_2^{(1)}}\tilde q)q
+6(\partial_{t_1^{(1)}}\tilde q)\partial_{t_1^{(1)}}q=0, \\
&4\partial_{t_3^{(2)}}q
-\partial_{t_1^{(2)}}^3q
+3\partial_{t_1^{(2)}t_2^{(2)}}^2q
-6q\partial_{t_2^{(2)}}\tilde r
+6(\partial_{t_1^{(2)}}q)\partial_{t_1^{(2)}}\tilde r=0, \\
&4\partial_{t_3^{(1)}}r
-\partial_{t_1^{(1)}}^3r
+3\partial_{t_1^{(1)}t_2^{(1)}}^2r
-6r\partial_{t_2^{(1)}}\tilde q
+6(\partial_{t_1^{(1)}}r)\partial_{t_1^{(1)}}\tilde q=0, \\
&4\partial_{t_3^{(2)}}r
-\partial_{t_1^{(2)}}^3r
-3\partial_{t_1^{(2)}t_2^{(2)}}^2r
+6(\partial_{t_2^{(2)}}\tilde r)r
+6(\partial_{t_1^{(2)}}\tilde r)
\partial_{t_1^{(2)}}r=0.
\end{split}
\end{equation}
\end{proposition}

\begin{proof}
We begin with the $\boldsymbol t^{(1)}$ derivatives of $q$.
Consider the linear combination
\begin{align*}
x_1\partial_{t_3^{(1)}}q+x_2\partial_{t_1^{(1)}}^3q+x_3\partial_{t_1^{(1)}t_2^{(1)}}^2q=\star,
\end{align*}
where $x_1,x_2,x_3$ are undetermined coefficients in $\mathbb{C}$, and $\star$ denotes the resulting (yet to be simplified) expression.
From the derivative formulas \eqref{third_q} in Appendix \ref{appendix:B}, each of the three derivatives
contributes the single terms $Q^{3,0}$ and $\tilde q\,Q^{2,0}$. Collecting them in $\star$ yields
\begin{align*}
(x_1+x_2+x_3)Q^{3,0}+(x_1+x_2+x_3)\tilde q\,Q^{2,0}.
\end{align*}
Hence their cancellation requires
\begin{align*}
x_1+x_2+x_3=0.
\end{align*}
Under this condition, $\star$ reduces to
\begin{align*}
\star={}&\bigl[(3x_2+x_3)\tilde Q^{1,0}+(x_1-2x_2)\tilde Q^{0,1}+(3x_2+x_3)\tilde q^2\bigr]Q^{1,0}\\
&+\bigl[(3x_2+x_3)\tilde Q^{2,0}+(-3x_2+x_3)\tilde Q^{1,1}+(x_1+x_2-x_3)\tilde Q^{0,2}\\
&+6x_2\tilde Q^{1,0}\tilde q+(-3x_2+x_3)\tilde Q^{0,1}\tilde q+(3x_2+x_3)\tilde q\tilde Q^{1,0}\\
&+(-3x_2+x_3)\tilde q\tilde Q^{0,1}+6x_2\tilde q^3\bigr]q.
\end{align*}

We now examine the structure of the residual quasi-Wronskian monomials on the right-hand side.
The monomials multiplying $Q^{1,0}$ are generated naturally by the product $(\partial_{t_1^{(1)}}\tilde q)\partial_{t_1^{(1)}}q$, since
\begin{align*}
\partial_{t_1^{(1)}}q=Q^{1,0}+\tilde q\,q,\qquad
\partial_{t_1^{(1)}}\tilde q=\tilde Q^{1,0}-\tilde Q^{0,1}+\tilde q^2.
\end{align*}
The monomials multiplying $q$ include $\tilde Q^{2,0}q$, $\tilde Q^{0,2}q$, and notably $\tilde Q^{1,1}q$. Therefore, both $(\partial_{t_2^{(1)}}\tilde q)q$ and $(\partial_{t_1^{(1)}}^2\tilde q)q$ may enter the required combination.
We therefore consider the most general combination of products of lower-order derivatives with undetermined coefficients in $\mathbb{C}$
\begin{align*}
z_1(\partial_{t_1^{(1)}}\tilde q)\partial_{t_1^{(1)}}q
+z_2(\partial_{t_1^{(1)}}^2\tilde q)q
+z_3(\partial_{t_2^{(1)}}\tilde q)q,
\end{align*}
where
\begin{align*}
&\partial_{t_2^{(1)}}\tilde q=\tilde Q^{2,0}-\tilde Q^{0,2}+\tilde q\,\tilde Q^{1,0}+\tilde Q^{0,1}\tilde q,\\
&\partial_{t_1^{(1)}}^2\tilde q=\tilde Q^{2,0}-2\tilde Q^{1,1}+\tilde Q^{0,2}
+(2\tilde Q^{1,0}-\tilde Q^{0,1})\tilde q+\tilde q(\tilde Q^{1,0}-2\tilde Q^{0,1})+2\tilde q^3.
\end{align*}
The coefficients of the product terms are summarized in Table~\ref{higher_Q_terms_3}.

\begin{table}[h]
\centering
\renewcommand{\arraystretch}{1.25}
\begin{tabular}{|c|c|c|c|c|}
\hline
 & $\star$ & $(\partial_{t_1^{(1)}}\tilde q)\partial_{t_1^{(1)}}q$ & $(\partial_{t_1^{(1)}}^2\tilde q)q$ & $(\partial_{t_2^{(1)}}\tilde q)q$ \\ \hline
$\tilde Q^{1,0}Q^{1,0}$ ($\tilde q^2Q^{1,0}$) & $3x_2+x_3$ & $1$ & $0$ & $0$ \\ \hline
$\tilde Q^{0,1}Q^{1,0}$ & $x_1-2x_2$ & $-1$ & $0$ & $0$ \\ \hline
$\tilde Q^{2,0}q$ ($\tilde q\,\tilde Q^{1,0}q$) & $3x_2+x_3$ & $0$ & $1$ & $1$ \\ \hline
$\tilde Q^{1,1}q$ ($\tilde q\,\tilde Q^{0,1}q$) & $-3x_2+x_3$ & $0$ & $-2$ & $0$ \\ \hline
$\tilde Q^{0,2}q$ & $x_1+x_2-x_3$ & $0$ & $1$ & $-1$ \\ \hline
$\tilde Q^{1,0}\tilde q\,q$ ($\tilde q^3q$) & $6x_2$ & $1$ & $2$ & $0$ \\ \hline
$\tilde Q^{0,1}\tilde q\,q$ & $-3x_2+x_3$ & $-1$ & $-1$ & $1$ \\ \hline
\end{tabular}
\caption{Coefficients of the product terms in the $\boldsymbol t^{(1)}$ derivatives of $q$.}
\label{higher_Q_terms_3}
\end{table}

Requiring
\begin{align*}
\star+z_1(\partial_{t_1^{(1)}}\tilde q)\partial_{t_1^{(1)}}q
+z_2(\partial_{t_1^{(1)}}^2\tilde q)q
+z_3(\partial_{t_2^{(1)}}\tilde q)q=0
\end{align*}
and comparing the coefficients row by row in Table~\ref{higher_Q_terms_3}, we obtain the following linear system
\begin{align*}
\left\{
\begin{array}{rcrcrcrcrcrcr}
x_1 & + & x_2 & + & x_3 & & & & & & & = & 0\\
    &   & 3x_2 & + & x_3 & + & z_1 & & & & & = & 0\\
    &   &      &   & 2x_3 & + & z_1 & - & 2z_2 & & & = & 0\\
    &   &      &   &      &   & z_1 & - & z_2 & - & z_3 & = & 0
\end{array}
\right.
\end{align*}
Solving this system gives
\begin{align*}
(x_1,x_2,x_3,z_1,z_2,z_3)
=\alpha(4,-1,-3,6,0,6)+\beta(0,1,-1,-2,-2,0).
\end{align*}
The term proportional to $\beta$ is precisely the $t_1^{(1)}$-derivative of the first equation in~\eqref{2cKP},
\begin{align*}
\partial_{t_1^{(1)}}\Bigl(\partial_{t_1^{(1)}}^2q-\partial_{t_2^{(1)}}q-2(\partial_{t_1^{(1)}}\tilde q)q\Bigr)
=\partial_{t_1^{(1)}}^3q-\partial_{t_1^{(1)}t_2^{(1)}}^2q-2(\partial_{t_1^{(1)}}^2\tilde q)q-2(\partial_{t_1^{(1)}}\tilde q)\partial_{t_1^{(1)}}q=0,
\end{align*}
hence it does not produce a genuinely new equation. Removing it by setting $\beta=0$ and fixing the normalization with $\alpha=1$, we obtain the first equation in \eqref{third_order_2cKP}.
Applying the same procedure to the remaining three cases, namely the $\boldsymbol{t}^{(2)}$-derivatives of $q$ and the $\boldsymbol{t}^{(1)}$- and $\boldsymbol{t}^{(2)}$-derivatives of $r$, and removing in each case the derivative of the corresponding second-order equation in \eqref{2cKP}, we obtain the remaining three equations in \eqref{third_order_2cKP}.
\end{proof}

In the commutative limit, the above system reduces to the following bilinear forms
\begin{equation}
\begin{split}
&(4\mathrm{D}_{t_3^{(1)}}-\mathrm{D}_{t_1^{(1)}}^3
-3\mathrm{D}_{t_1^{(1)}}\mathrm{D}_{t_2^{(1)}})
\tau_{n+1,m-1}\cdot\tau_{n,m}=0, \\
&(4\mathrm{D}_{t_3^{(2)}}-\mathrm{D}_{t_1^{(2)}}^3
+3\mathrm{D}_{t_1^{(2)}}\mathrm{D}_{t_2^{(2)}})
\tau_{n+1,m-1}\cdot\tau_{n,m}=0, \\
&(4\mathrm{D}_{t_3^{(1)}}-\mathrm{D}_{t_1^{(1)}}^3
+3\mathrm{D}_{t_1^{(1)}}\mathrm{D}_{t_2^{(1)}})
\tau_{n-1,m+1}\cdot\tau_{n,m}=0, \\
&(4\mathrm{D}_{t_3^{(2)}}-\mathrm{D}_{t_1^{(2)}}^3
-3\mathrm{D}_{t_1^{(2)}}\mathrm{D}_{t_2^{(2)}})
\tau_{n-1,m+1}\cdot\tau_{n,m}=0.
\end{split}
\end{equation}

\subsection{Matrix-valued solutions of the non-commutative two-component KP hierarchy} The preceding construction provides solutions of the non-commutative two-component KP hierarchy in terms of quasi-Wronskians. In this part, we demonstrate how to choose appropriate seed functions so as to obtain matrix-valued solutions. To this end, we take the seed functions $\phi_i$ and $\psi_i$ as $d \times d$ matrix-valued functions of the form
\begin{equation} \label{matrixseed}
\begin{split}
\phi_i = \frac{1-(-1)^i}{2}I_d + r_i P_i e^{\eta_i}, \\
\psi_i = \frac{1+(-1)^i}{2}I_d + v_i P_i e^{\gamma_i},
\end{split} 
\end{equation}
where
\begin{align*} 
\eta_i=k_ix+\sum_{j=1}^{\infty} k_i^j t_j^{(1)} + \eta_i^0, \qquad \gamma_i=q_ix+\sum_{j=1}^{\infty} q_i^j t_j^{(2)} + \gamma_i^0, 
\end{align*}
with $r_i$, $v_i$, $k_i$, $q_i$, $\eta_i^0$, $\gamma_i^0$ being scalars, and $I_d$ denoting the $d\times d$ identity matrix. Moreover, $P_i$ are the rank-$1$ projection matrices given by
\begin{equation}
P_i=\frac{\mu_i\otimes\nu_i}{(\mu_i,\nu_i)}
=\frac{\mu_i\nu_i^{\mathsf T}}{\mu_i^{\mathsf T}\nu_i},
\end{equation}
where $\mu_i,\nu_i\in\mathbb{C}^d$ and
$\mu_i^{\mathsf T}\nu_i\neq 0$. 

Specializing to the case $n=m=1$, we have
\begin{equation}
\begin{split}
q &= \left|\begin{array}{ccc}
\phi_1&\phi_2&0_d\\
\psi_1&\psi_2&I_d\\
\phi_1'&\phi_2'& \boxed{0_d}
\end{array}
\right| 
=-\begin{pmatrix} \phi_1' & \phi_2'\end{pmatrix}\begin{pmatrix} \phi_1&\phi_2 \\ \psi_1&\psi_2\end{pmatrix}^{-1} \begin{pmatrix} 0_d \\ I_d \end{pmatrix} \\
& = - \begin{pmatrix} K_1 & K_2 \end{pmatrix} \begin{pmatrix} 0_d \\ I_d \end{pmatrix} = - K_2,
\end{split}
\end{equation}
where $0_d$ is the $d \times d$ zero matrix, and $K_1$, $K_2$ satisfy
\begin{equation}
\begin{split}
&K_1 \phi_1 + K_2 \psi_1 = \phi_1', \\
&K_1 \phi_2 + K_2 \psi_2= \phi_2'.
\end{split}
\end{equation}
Solving for $K_2$ gives
\begin{align}
\begin{aligned}
q=-K_2 = -\frac{\alpha_2[k_2(1+\alpha_1)I_d-k_1\alpha_1P_1]P_2}{(1+\alpha_1)(1+\beta_2)-\epsilon \alpha_2 \beta_1},
\end{aligned}
\end{align}
where $\alpha_i=r_i e^{\eta_i}$, $\beta_i=v_i e^{\gamma_i}$, and
\begin{equation} \label{rank-one-epsilon}
\epsilon=
\frac{(\nu_1^{\mathsf T}\mu_2)(\nu_2^{\mathsf T}\mu_1)}
{(\mu_1^{\mathsf T}\nu_1)(\mu_2^{\mathsf T}\nu_2)}.
\end{equation}

Analogously, the expression for $r$ is obtained as
\begin{equation}
r = -\frac{\beta_1[q_1(1+\beta_2)I_d-q_2\beta_2P_2]P_1}{(1+\alpha_1)(1+\beta_2)-\epsilon \alpha_2 \beta_1}.
\end{equation}

\section{Reduction techniques in the direct method} \label{sec:5}

In this section, we impose reductions directly on the seed functions and investigate how these induce corresponding changes in the quasi-Wronskian derivation rules. We first reduce the seed functions of the non-commutative two-component KP hierarchy to those of the one-component counterpart. Subsequently, we perform a stepwise reduction on the seed functions associated with the non-commutative two-component KP equation, which yields the non-commutative DS and NLS equations.

\subsection{Reduction to the non-commutative one-component KP hierarchy}
\label{sec:5.1}

Recall that the seed functions of the non-commutative two-component KP hierarchy
satisfy the dispersion relations \eqref{dispersion}.  The reduction to the non-commutative one-component KP hierarchy is achieved by setting $m=0$.  Under this reduction, the $\bpsi$-family vanishes, leaving only the seed functions $\phi_1,\ldots,\phi_n$.  We retain the first family of flows, denoting $t_p=t_p^{(1)}$ and $\partial_{t_p}=\partial_{t_p^{(1)}}$ for $p\geq 1$. The surviving seed functions therefore satisfy
\begin{align*}
\partial_{t_p}\phi_j=\partial_x^p\phi_j,
\qquad
j=1,\ldots,n,\quad p\geq 1,
\end{align*}
while the second family of flows $\partial_{t_p^{(2)}}$ is discarded.

When $m=0$, the two-component Wronski matrix \eqref{2c_wronski} reduces to the ordinary Wronski matrix
\[
W_n(\bphi)=W(\phi_1,\ldots,\phi_n),
\]
and among the generalized quasi-Wronskians introduced in
Definition~\ref{quasi-Wronskians}, only the family
$\tilde Q^{i,k}$ remains, namely
\begin{equation}
\tilde Q^{i,k}
=
\left|
\begin{array}{cc}
W_n(\boldsymbol{\phi}) & e_{n-k}\\
\partial^{n+i}\boldsymbol{\phi} & \boxed{0}
\end{array}
\right|.
\label{single-component-quasi-wronskian}
\end{equation}
These are precisely the one-component quasi-Wronskians.

Recall that, in \cite{li26}, we have the evolution rule
\begin{align*}
\partial_{t_p}\tilde{Q}^{i,k}=\tilde{Q}^{i+p,k}-\tilde{Q}^{i,k+p}+\sum_{a=0}^{p-1}\tilde{Q}^{i,a}\tilde{Q}^{p-a-1,k},\quad p\geq 1,
\end{align*}
which coincide with the $\partial_{t_p^{(1)}}$-derivation rule for $\tilde{Q}^{i,k}$ in Proposition \ref{prop1}. For further details on the resulting non-commutative one-component KP hierarchy, including its higher flows and associated quasi-Wronskian solutions, we refer the reader to \cite{li26}.

\subsection{Reduction to the non-commutative DS equation} 
Recall that the seed functions of the non-commutative two-component KP
equation \eqref{2cKP} satisfying the dispersion relations
\eqref{seed-dispersion} can be taken as
\begin{equation}
\begin{split}
\phi_j &= a_j\exp\bigl(k_j (t_1^{(1)}+x)+k_j^2 t_2^{(1)}\bigr)
+b_j\exp\bigl(l_j (t_1^{(1)}+x)+l_j^2 t_2^{(1)}\bigr), \\
\psi_j &= c_j\exp\bigl(\alpha_j (t_1^{(2)}+x)+\alpha_j^2 t_2^{(2)}\bigr)
+d_j\exp\bigl(\beta_j (t_1^{(2)}+x)+\beta_j^2 t_2^{(2)}\bigr),
\end{split}
\label{kp_seed_exp}
\end{equation}
where $a_j,b_j,c_j,d_j,k_j,l_j,\alpha_j,\beta_j$ are elements of a (not
necessarily commutative) unital associative algebra $\mathcal A$. We now introduce
new variables $y,t$ by
\begin{equation}
\label{reduction}
t_1^{(1)}=-\sigma^{-1}y,\qquad
t_1^{(2)}=\sigma^{-1}y,\qquad
t_2^{(1)}=2\mathrm{i}t,\qquad
t_2^{(2)}=-2\mathrm{i}t,
\end{equation}
where $\sigma=-1$ or $\sigma=\mathrm{i}$. Under this substitution the seed
functions \eqref{kp_seed_exp} become
\begin{equation}
\begin{split}
\phi_j &= a_j\exp\left(k_j(x-\sigma^{-1}y)+2\mathrm{i}k_j^2 t\right)
+b_j\exp\left(l_j(x-\sigma^{-1}y)+2\mathrm{i}l_j^2 t\right), \\
\psi_j &= c_j\exp\left(\alpha_j(x+\sigma^{-1}y)-2\mathrm{i}\alpha_j^2 t\right)
+d_j\exp\left(\beta_j(x+\sigma^{-1}y)-2\mathrm{i}\beta_j^2 t\right),
\end{split}
\label{ds_seed_exp}
\end{equation}
and one checks directly that they satisfy the
dispersion relations \eqref{seed-dispersion-ds} used in
Section~\ref{sec:3.3}.

The dispersion relations \eqref{seed-dispersion-ds} specialize the derivative
rules of Proposition~\ref{prop1} as follows.

\begin{proposition} \label{prop2}
If the seed functions $\phi_1, \phi_2, \ldots, \phi_{n+m}$ and $\psi_1, \psi_2, \ldots, \psi_{n+m}$ satisfy the dispersion relations \eqref{seed-dispersion-ds}, then the quasi-Wronskians $Q^{i,j}$, $\tilde{Q}^{i,k}$, $\tilde R^{i,j}$ and $R^{i,k}$ satisfy the following derivative rules
\begin{equation}
\begin{split}
\partial_xQ^{i,j} &= Q^{i+1,j} - Q^{i,j+1} + Q^{i,0}\tilde{R}^{0,j} + \tilde{Q}^{i,0}Q^{0,j}, \\  
\partial_yQ^{i,j} &= -\sigma^{-1}
(Q^{i+1,j} + Q^{i,j+1} - Q^{i,0}\tilde{R}^{0,j} + \tilde{Q}^{i,0}Q^{0,j}), \\ 
\partial_tQ^{i,j} &= 2\mathrm{i}(Q^{i+2,j} + Q^{i,j+2} - Q^{i,0}\tilde{R}^{1,j} - Q^{i,1}\tilde{R}^{0,j} + \tilde{Q}^{i,0}Q^{1,j} + \tilde{Q}^{i,1}Q^{0,j}), \\
\partial_x\tilde{Q}^{i,k} &= \tilde{Q}^{i+1,k} - \tilde{Q}^{i,k+1} + Q^{i,0} R^{0,k} + \tilde{Q}^{i,0} \tilde{Q}^{0,k}, \\
\partial_y\tilde{Q}^{i,k} &= -\sigma^{-1}(\tilde{Q}^{i+1,k} - \tilde{Q}^{i,k+1} - Q^{i,0} R^{0,k} + \tilde{Q}^{i,0} \tilde{Q}^{0,k}), \\
\partial_x\tilde{R}^{i,j} &= \tilde{R}^{i+1,j} - \tilde{R}^{i,j+1} + R^{i,0} Q^{0,j} + \tilde{R}^{i,0} \tilde{R}^{0,j}, \\
\partial_y\tilde{R}^{i,j} &= \sigma^{-1}(\tilde{R}^{i+1,j} - \tilde{R}^{i,j+1} - R^{i,0} Q^{0,j} + \tilde{R}^{i,0} \tilde{R}^{0,j}), \\
\partial_xR^{i,k} &= R^{i+1,k} - R^{i,k+1} + R^{i,0}\tilde{Q}^{0,k} + \tilde{R}^{i,0}R^{0,k}, \\  
\partial_yR^{i,k} &= \sigma^{-1}
(R^{i+1,k} + R^{i,k+1} - R^{i,0}\tilde{Q}^{0,k} + \tilde{R}^{i,0}R^{0,k}), \\ 
\partial_tR^{i,k} &= -2\mathrm{i}(R^{i+2,k} + R^{i,k+2} - R^{i,0}\tilde{Q}^{1,k} - R^{i,1}\tilde{Q}^{0,k} + \tilde{R}^{i,0}R^{1,k} + \tilde{R}^{i,1}R^{0,k}).
\end{split}
\end{equation}
\end{proposition}

In particular, Proposition \ref{prop2} gives
\begin{align} \label{ncds_d}
\begin{aligned}
&\partial_xq = Q^{1,0} - Q^{0,1} + q \tilde{r} + \tilde{q} q,
\qquad \partial_yq = -\sigma^{-1}(Q^{1,0} + Q^{0,1} - q \tilde{r} + \tilde{q} q), \\
&\partial_x\tilde q=\tilde Q^{1,0}-\tilde Q^{0,1}+qr+\tilde q^2,
\qquad \partial_y\tilde q=-\sigma^{-1}(\tilde Q^{1,0}-\tilde Q^{0,1}-qr+\tilde q^2), \\
&\partial_xr = R^{1,0} - R^{0,1} + r \tilde{q} + \tilde{r} r,
\qquad \partial_yr = \sigma^{-1}(R^{1,0} + R^{0,1} - r \tilde{q} + \tilde{r} r), \\
&\partial_x\tilde r=\tilde R^{1,0}-\tilde R^{0,1}+rq+\tilde r^2,
\qquad \partial_y\tilde r=\sigma^{-1}(\tilde R^{1,0}-\tilde R^{0,1}-rq+\tilde r^2), \\
&\partial_tq = 2\mathrm{i}(Q^{2,0} + Q^{0,2} - q \tilde{R}^{1,0} - Q^{0,1}\tilde{r} + \tilde{q} Q^{1,0} + \tilde{Q}^{0,1} q), \\
&\partial_tr = -2\mathrm{i}(R^{2,0} + R^{0,2} - r \tilde{Q}^{1,0} - R^{0,1}\tilde{q} + \tilde{r} R^{1,0} + \tilde{R}^{0,1} r), \\
&\partial_x^2q = Q^{2,0} - 2 Q^{1,1} + Q^{0,2} + (2Q^{1,0}-Q^{0,1})\tilde{r} + q (\tilde{R}^{1,0}-2\tilde{R}^{0,1}) \\
&\quad+ (2\tilde{Q}^{1,0}-\tilde{Q}^{0,1}) q + \tilde{q} (Q^{1,0}-2Q^{0,1}) + 2q \tilde{r}^{2} + 2\tilde{q} q \tilde{r} + 2q r q + 2 \tilde{q}^2 q, \\
&\partial_y^2q = \sigma^{-2}[Q^{2,0} + 2 Q^{1,1} + Q^{0,2} - (2Q^{1,0}+Q^{0,1})\tilde{r} + q (\tilde{R}^{1,0}-2\tilde{R}^{0,1})  \\
&\quad+ (2\tilde{Q}^{1,0}-\tilde{Q}^{0,1}) q + \tilde{q} (Q^{1,0}+2Q^{0,1}) + 2q \tilde{r}^{2} - 2\tilde{q} q \tilde{r} - 2q r q + 2\tilde{q}^2 q], \\
&\partial_x^2r = R^{2,0} - 2 R^{1,1} + R^{0,2} + (2R^{1,0}-R^{0,1})\tilde{q} + r (\tilde{Q}^{1,0}-2\tilde{Q}^{0,1}) \\
&\quad+ (2\tilde{R}^{1,0}-\tilde{R}^{0,1}) r + \tilde{r} (R^{1,0}-2R^{0,1}) + 2r \tilde{q}^{2} + 2\tilde{r} r \tilde{q} + 2r q r + 2 \tilde{r}^2 r, \\
&\partial_y^2r = \sigma^{-2}[R^{2,0} + 2 R^{1,1} + R^{0,2} - (2R^{1,0}+R^{0,1})\tilde{q} + r (\tilde{Q}^{1,0}-2\tilde{Q}^{0,1})  \\
&\quad+ (2\tilde{R}^{1,0}-\tilde{R}^{0,1}) r + \tilde{r} (R^{1,0}+2R^{0,1}) + 2r \tilde{q}^{2} - 2\tilde{r} r \tilde{q} - 2r q r + 2\tilde{r}^2 r].
\end{aligned}
\end{align}

Eliminating the quasi-Wronskian monomials $Q^{i,j}$, $\tilde{Q}^{i,j}$, $\tilde R^{i,j}$ and $R^{i,j}$ with $(i,j)\neq(0,0)$ from
\eqref{ncds_d} by suitable linear combinations yields the reduced system. For instance, the coefficients of the monomials
$Q^{i,j}$ with $i+j=2$ appearing in $\partial_t q$, $\partial_x^2 q$ and
$\partial_y^2 q$ are listed in Table~\ref{Q_terms_5}.
\begin{table}[h]
    \centering
    \renewcommand{\arraystretch}{1.25}
    \begin{tabular}{|c|c|c|c|}
        \hline
        & $\partial_tq$ & $\partial_x^2q$ & $\partial_y^2q$ \\ \hline
        $Q^{2, 0}(Q^{0, 2})$ & $2\mathrm{i}$ & $1$ & $\sigma^{-2}$ \\ \hline
        $Q^{1, 1}$ & $0$ & $-2$ & $2\sigma^{-2}$ \\ \hline
    \end{tabular}
    \caption{Coefficients of the monomials $Q^{i,j}$ with $i+j=2$ in the
    derivatives of $q$.}
    \label{Q_terms_5}
\end{table}

Consequently, the combination
\[
\mathrm{i}\partial_tq+\partial_x^2q+\sigma^2\partial_y^2q
\]
eliminates all monomials $Q^{i,j}$ with $i+j=2$, and the remaining terms give
\begin{equation} \label{mid_step}
\begin{split}
\mathrm{i}\partial_tq+\partial_x^2q+\sigma^2\partial_y^2q &= 4(\tilde{Q}^{1, 0}-\tilde{Q}^{0, 1}+\tilde{q}^2)q + 4q(\tilde{R}^{1, 0}-\tilde{R}^{0, 1}+\tilde{r}^{2}).
\end{split}
\end{equation}
Combining the $x$- and $y$-derivative identities of $\tilde q$ and $\tilde r$ in
\eqref{ncds_d}, we obtain
\begin{align}
\begin{aligned}
&\partial_x\tilde q-\sigma\partial_y\tilde q=2(\tilde Q^{1,0}-\tilde Q^{0,1}+\tilde q^2),
&&\partial_x\tilde q+\sigma\partial_y\tilde q=2qr, \\
&\partial_x\tilde r+\sigma\partial_y\tilde r=2(\tilde R^{1,0}-\tilde R^{0,1}+\tilde r^2),
&&\partial_x\tilde r-\sigma\partial_y\tilde r=2rq.
\end{aligned}
\end{align}
Therefore, \eqref{mid_step} is precisely the first equation of
the following system, while its third and fourth equations follow from the
last two identities above. The second equation is obtained in the
same way from the corresponding identities for $r$. Hence we arrive at the non-commutative DS system
\begin{equation} \label{ncds}
\begin{split}
&\mathrm{i} \partial_tq + (\partial_x^2 + \sigma^2 \partial_y^2)q
- 2(\partial_x-\sigma\partial_y)\tilde q\,q
- 2q(\partial_x+\sigma\partial_y)\tilde r = 0, \\
&\mathrm{i} \partial_tr - (\partial_x^2 + \sigma^2 \partial_y^2)r
+ 2(\partial_x+\sigma\partial_y)\tilde r\,r
+ 2r(\partial_x-\sigma\partial_y)\tilde q = 0, \\
&(\partial_x+\sigma\partial_y)\tilde{q}-2qr=0, \\
&(\partial_x-\sigma\partial_y)\tilde{r}-2rq=0.
\end{split}
\end{equation}
Using the two constraint equations to eliminate $\partial_x\tilde q$ and
$\partial_x\tilde r$, the first two equations of \eqref{ncds} take the form
\eqref{ncds_marchenko} obtained by the involution reduced twisted Marchenko lemma in
Section~\ref{sec:3.3}. Note also that no $t$-derivative formulas for
$\tilde Q^{i,k}$ and $\tilde R^{i,j}$ are needed, since the $t$-derivatives of
the auxiliary fields $\tilde q$ and $\tilde r$ do not occur in \eqref{ncds}.

\begin{remark}
Gilson and Macfarlane~\cite{gilson09} have constructed quasi-Wronskian and quasi-Grammian solutions of the non-commutative DS system via Darboux and binary Darboux transformations, respectively. For the purpose of obtaining matrix-valued dromion solutions, they adopt the quasi-Grammian formulation in the DSI case, which corresponds to $\sigma=-1$. The seed functions are chosen as
\begin{align*}
\phi_j=e^{p_jx+\mathrm{i}p_j^2t+p_jy+\alpha_j^0} I_2,\qquad
\psi_j=e^{q_jx-\mathrm{i}q_j^2t-q_jy+\beta_j^0} I_2.
\end{align*}
where $p_j$ and $q_j$ are constants with positive real parts, and
$\alpha_j^0$ and $\beta_j^0$ are phase constants. Thus, the seed
functions themselves are scalar exponentials multiplied by the identity matrix $I_2$,
while the matrix-valued structure is encoded in the constant
invertible Hermitian matrix $H$ that appears in the quasi-Grammian
construction.
\end{remark}

\subsection{Reduction to the non-commutative NLS equation.}
If $q$, $\tilde{r}$, $\tilde{q}$, and $r$ are independent of the $y$-variable, then their $y$-derivatives vanish identically. Under this reduction, the non-commutative DS equation~\eqref{ncds} consistently reduces to the non-commutative NLS equation~\cite{dimakis01}, which takes the form
\begin{equation} \label{ncnls}
\begin{split}
\mathrm{i}\partial_tq + \partial_x^2q - 8qrq &= 0, \\
\mathrm{i}\partial_tr - \partial_x^2r + 8rqr &= 0. \\
\end{split}
\end{equation}
Therefore, in order to obtain solutions of the non-commutative NLS equation~\eqref{ncnls} as well as the non-commutative DS equation~\eqref{ncds} within the present framework, it is necessary to impose suitable constraints on the seed functions so that the $y$-derivatives of $q$, $\tilde r$, $\tilde q$, and $r$ vanish identically. This observation leads to the following theorem.

\begin{theorem} \label{thm3}
If the seed functions $\phi_1,\ldots,\phi_{n+m}$ and $\psi_1,\ldots,\psi_{n+m}$ satisfy both the dispersion relations \eqref{seed-dispersion-ds}
and the constraints 
\begin{equation}
\partial_x \phi_j=-\sigma\phi_j k_j, \qquad \partial_x \psi_j = \sigma\psi_j k_j,
\end{equation} 
then $\partial_yq= \partial_y\tilde{r}=\partial_y\tilde{q}=\partial_yr=0$.
Here $k_j$ is a right multiplier in a (not necessarily commutative) unital associative algebra. 
\end{theorem}

\begin{proof}
By applying the derivative formula for quasi-determinants, we obtain
\begin{equation} \label{y-derivative}
\begin{split}
\partial_yq = - (\partial_y \boldsymbol{\phi}^{(n)}) W_{n,m}(\bphi,\bpsi)^{-1} e_{n+m} + \boldsymbol{\phi}^{(n)} W_{n,m}(\bphi,\bpsi)^{-1} (\partial_y W_{n,m}(\bphi,\bpsi)) W_{n,m}(\bphi,\bpsi)^{-1} e_{n+m}.
\end{split}
\end{equation}
Since $\partial_y \phi_j=-\sigma^{-1}\partial_x\phi_j = \phi_j k_j$ and $\partial_y \psi_j=\sigma^{-1}\partial_x\psi_j=\psi_j k_j$, it follows that
\[
\partial_y \boldsymbol{\phi}^{(n)} = \boldsymbol{\phi}^{(n)} \operatorname{diag}(k_1,k_2,\cdots,k_{n+m}), \qquad \partial_y W_{n,m}(\bphi,\bpsi) = W_{n,m}(\bphi,\bpsi) \operatorname{diag}(k_1,k_2,\cdots,k_{n+m}).    
\]
Substituting these expressions into \eqref{y-derivative} immediately yields $\partial_yq = 0$. The same argument applied to $\tilde r$, $\tilde q$, and $r$ gives $\partial_y\tilde r=\partial_y\tilde q=\partial_yr=0$.
\end{proof}

We now turn to the construction of a matrix-valued solution to \eqref{ncnls} for the case $n=m=1$.
Following the matrix-valued seed functions introduced for the non-commutative two-component KP hierarchy, we consider
\begin{equation}\label{matrix-nls-seeds}
\begin{aligned}
\phi_j&=\left(\frac{1-(-1)^j}{2}I_d+r_jP_j\right)e^{\eta_j},\\
\psi_j&=\left(\frac{1+(-1)^j}{2}I_d+v_jP_j\right)e^{\gamma_j},
\end{aligned}
\end{equation}
where $\eta_j=-\sigma k_jx+2\mathrm i \sigma^2k_j^2t+\eta_j^0$, $\gamma_j=\sigma k_jx-2\mathrm i\sigma^2k_j^2t+\gamma_j^0$, $k_j$ is a scalar, and $r_j,v_j,P_j$ are the same as in \eqref{matrixseed}.

Expanding the quasi-determinant expression yields
\begin{equation}
    \begin{split}
    q &= \sigma\begin{pmatrix} \phi_1k_1& \phi_2k_2 \end{pmatrix}  \begin{pmatrix} \phi_1&\phi_2 \\\psi_1&\psi_2\end{pmatrix}^{-1} \begin{pmatrix} 0_d \\ I_d \end{pmatrix} = \sigma \begin{pmatrix} S_1 & S_2 \end{pmatrix} \begin{pmatrix} 0_d \\ I_d \end{pmatrix} = \sigma S_2,
    \end{split}
\end{equation}
where $S_1$ and $S_2$ satisfy
\begin{equation}
    \begin{split}
        &S_1\phi_1+S_2\psi_1=\phi_1k_1, \\
        &S_1\phi_2+S_2\psi_2=\phi_2k_2.
    \end{split}
\end{equation}
Solving for $S_2$, we obtain
\begin{equation}
q=\sigma S_2=\frac{\sigma (k_2-k_1)r_2(1+r_1)P_2}{(1+r_1)(1+v_2)e^{\gamma_2-\eta_2}- \epsilon r_2v_1e^{\gamma_1-\eta_1}},
\end{equation}
where $\epsilon$ is given by \eqref{rank-one-epsilon}. Similarly, the expression for $r$ is found to be
\begin{equation}
r=\frac{\sigma (k_2-k_1)v_1(1+v_2)P_1}{(1+r_1)(1+v_2)e^{\eta_1-\gamma_1}- \epsilon r_2v_1e^{\eta_2-\gamma_2}} .
\end{equation}

\section{Concluding remarks}
Non-commutative KP theory, together with its discrete analogy, have been widely studied in recent years via different methods \cite{sakakibara04,dimakis07,gilson07,he11,li18,ma10,dimakis04,dimakis05,kupershmidt00}. Doliwa's recent work further supports this direction. Non-commutative Hermite-Pad\'e approximation
leads to quasi-determinantal solutions of the non-commutative Hirota (discrete KP) system and, through suitable reductions, to non-commutative discrete-time Toda equations \cite{doliwa22}. Related studies connect multiple interpolation and multiple orthogonal polynomials with multidimensional Toda systems and determinantal identities \cite{doliwa22_,doliwa24}.

In this work, we used a novel perspective by using the Marchenko lemma. From the linear problem satisfied by moment matrices, we constructed non-commutative integrable systems. A direct approach and reduction technique are also presented as a tool to generate non-commutative integrable systems from quasi-determinants. The non-triviality of the non-commutative two-component KP equation is demonstrated because the non-commutative KP and DS equations are the special reductions. Some connections with non-commutative Grassmannian and quasi-symmetric functions will be discussed in the future works.

\section*{Acknowledgements}
This work is partially funded by the National Natural Science Foundation of China (Grant No.~12671303).

\section*{Data Availability Statement}
There is no data needed in this research.

\section*{Interest Statement}
There is no conflict of interest.

\appendix
\section{Introduction to Quasi-determinants}
\noindent In the appendix, we briefly recall several fundamental properties of quasi-determinants. A detailed exposition can be found in \cite{gelfand91,gelfand92,gelfand05}.
\begin{definition}
Let $A = (a_{ij})$ be an $n \times n$ matrix with entries in a non-commutative unital ring $\mathcal{A}$. The $(i, j)^{\mathrm{th}}$ quasi-determinant of $A$, denoted by $\left| A \right|_{ij}$, is defined by
\[
\left| A \right|_{ij} = 
\begin{vmatrix}
A^{ij} & c_j^i \\
r_i^j & \boxed{a_{ij}}
\end{vmatrix} 
= a_{ij} - r_i^j (A^{ij})^{-1} c_j^i.
\]
Here, $A^{ij}$ denotes the $(n-1) \times (n-1)$ minor obtained from $A$ by removing its $i^{\mathrm{th}}$ row and $j^{\mathrm{th}}$ column, which is assumed to be invertible, $r_i^j$ is the row vector obtained from the $i^{\mathrm{th}}$ row of $A$ by deleting the $j^{\mathrm{th}}$ entry, and $c_j^i$ is the column vector obtained from the $j^{\mathrm{th}}$ column of $A$ by deleting the $i^{\mathrm{th}}$ entry. 
\end{definition}
If the entries in $A$ commute, one has
\[
\left| A \right|_{ij} = (-1)^{i+j} \frac{\det(A)}{\det(A^{ij})}.
\]
Thus quasi-determinants may be regarded as non-commutative analogues of the ratio of a determinant to one of its $(n-1)\times(n-1)$ minors.

We now recall from \cite{gelfand91,gelfand92,gelfand05} the behavior of quasi-determinants under elementary row and column operations.
\begin{proposition} \label{permutation}
A simultaneous permutation of the rows and columns of a quasi-determinant does not change its value.
\end{proposition}

\begin{proposition} \label{multiplying}
Suppose that the matrix $C$ is obtained from $A$ by multiplying the $i^{\mathrm{th}}$ row on the left by $\lambda$, then
\[
|C|_{kj} =
\begin{cases}
\lambda |A|_{ij}, & \text{for } k = i, \\[4pt]
|A|_{kj}, & \text{for } k \ne i.
\end{cases}
\]
Similarly, suppose that the matrix $C$ is obtained from $A$ by multiplying the $j^{\mathrm{th}}$ column on the right by $\mu$, then
\[
|C|_{il} =
\begin{cases}
|A|_{ij} \mu, & \text{for } l = j, \\[4pt]
|A|_{il}, & \text{for } l \ne j.
\end{cases}
\]
Finally, suppose that the matrix $D$ is obtained from $A$ by adding to one of its rows \textup{(}columns\textup{)} a multiple of the $k^{\mathrm{th}}$ row \textup{(}column\textup{)}, then
\[
|D|_{ij} = |A|_{ij} \quad \text{for every $i \ne k$ $\textup{(}j \ne k\textup{)}$}.
\]
\end{proposition}

The derivative of a quasi-determinant follows two distinct formulas, depending on the structure of the underlying matrix. For a general quasi-determinant, the derivative takes the form
\begin{equation} \nonumber
    \begin{vmatrix}
        A & B \\
        C & \boxed{d}
    \end{vmatrix}^{\prime} = d' - C'A^{-1}B + CA^{-1}A'A^{-1}B - CA^{-1}B',
\end{equation}
where $A$ is an $n \times n$ matrix, $C$ is a row vector and $B$ a column vector. 
This formula will serve as the basis for deriving more explicit expressions in the special cases considered below.
\begin{proposition} \label{derivative}
If $A$ is a Grammian-like matrix with derivative
\begin{equation} \nonumber
    A^{\prime} = \sum_{i=1}^k E_i F_i
\end{equation}
for some integer $k$, where $E_i$ \textup{(}$F_i$\textup{)} are column \textup{(}row\textup{)} vectors of appropriate length, then the derivative becomes
\begin{equation} \label{dquasigrammian}
    \begin{vmatrix}
        A & B \\
        C & \boxed{d}
    \end{vmatrix}^{\prime} = 
    \begin{vmatrix} 
        A & B \\ 
        C^{\prime} & \boxed{d^{\prime}} 
    \end{vmatrix} + 
    \begin{vmatrix} 
        A & B^{\prime} \\ 
        C & \boxed{0} 
    \end{vmatrix} + \sum_{i=1}^{k} 
    \begin{vmatrix} 
        A & E_i \\ 
        C & \boxed{0} 
    \end{vmatrix} 
    \begin{vmatrix} 
        A & B \\ 
        F_i & \boxed{0} 
    \end{vmatrix}.
\end{equation}
If the matrix $A$ does not have Grammian-like structure, then the derivative becomes
\begin{equation} \label{dquasiwronskian}
    \begin{split}
    \begin{vmatrix}
        A & B \\
        C & \boxed{d}
    \end{vmatrix}^{\prime} 
    & = 
    \begin{vmatrix}
        A & B \\
        C^{\prime} & \boxed{d^{\prime}}
    \end{vmatrix} 
        + \sum_{k=1}^{n} \begin{vmatrix}
        A & e_k \\
        C & \boxed{0}
    \end{vmatrix} 
    \begin{vmatrix}
        A & B \\
        (A^k)^{\prime} & \boxed{(B^k)^{\prime}}
    \end{vmatrix} \\
    & =
    \begin{vmatrix}
        A & B^{\prime} \\
        C & \boxed{d^{\prime}}
    \end{vmatrix} 
        + \sum_{k=1}^{n} 
    \begin{vmatrix}
        A & (A_k)^{\prime} \\
        C & \boxed{(C_k)^{\prime}}
    \end{vmatrix} 
    \begin{vmatrix}
        A & B \\
        e_k^\top & \boxed{0}
    \end{vmatrix},
\end{split}
\end{equation}
where $A^k$ \textup{(}$A_k$\textup{)} denotes the $k^\mathrm{th}$ row \textup{(}column\textup{)} of the matrix $A$.    
\end{proposition}

\section{Derivative formulas} \label{appendix:B}
This appendix lists the derivative formulas for $q$, $\tilde q$, $\tilde r$, and $r$ required for the derivation of the third-order non-commutative two-component KP equation.

1st order
\begin{equation}
\begin{split}
\partial_{t_1^{(1)}}q
={}&Q^{1,0}+\tilde qq,\\
\partial_{t_1^{(2)}}q
={}&-Q^{0,1}+q\tilde r,\\
\partial_{t_1^{(1)}}r
={}&-R^{0,1}+r\tilde q,\\
\partial_{t_1^{(2)}}r
={}&R^{1,0}+\tilde rr, \\
\partial_{t_1^{(1)}}\tilde q
={}&\tilde Q^{1,0}-\tilde Q^{0,1}
+\tilde q^2, \\
\partial_{t_1^{(2)}}\tilde r
={}&\tilde R^{1,0}-\tilde R^{0,1}
+\tilde r^2.
\end{split}
\end{equation}

2nd order
\begin{equation}
\begin{split}
\partial_{t_2^{(1)}}\tilde q
={}&\tilde Q^{2,0}-\tilde Q^{0,2}
+\tilde q\tilde Q^{1,0}
+\tilde Q^{0,1}\tilde q, \\
\partial_{t_1^{(1)}}^2\tilde q
={}&\tilde Q^{2,0}-2\tilde Q^{1,1}+\tilde Q^{0,2}
+(2\tilde Q^{1,0}
-\tilde Q^{0,1})\tilde q
+\tilde q(\tilde Q^{1,0}-2\tilde Q^{0,1})
+2\tilde q^3, \\
\partial_{t_2^{(2)}}\tilde r
={}&\tilde R^{2,0}-\tilde R^{0,2}
+\tilde r\tilde R^{1,0}
+\tilde R^{0,1}\tilde r, \\
\partial_{t_1^{(2)}}^2\tilde r
={}&\tilde R^{2,0}-2\tilde R^{1,1}+\tilde R^{0,2}
+(2\tilde R^{1,0}
-\tilde R^{0,1})\tilde r
+\tilde r(\tilde R^{1,0}
-2\tilde R^{0,1})
+2\tilde r^3.
\end{split}
\end{equation}

3rd order
\begin{equation} \label{third_q}
\begin{split}
\partial_{t_3^{(1)}}q
={}&Q^{3,0}+\tilde qQ^{2,0}+\tilde Q^{0,1}Q^{1,0}+\tilde Q^{0,2}q, \\
\partial_{t_1^{(1)}t_2^{(1)}}^2q
={}&Q^{3,0}
+\tilde qQ^{2,0}
+(\tilde Q^{1,0}
+\tilde q^2)Q^{1,0}\\
&+(
\tilde Q^{2,0}
+\tilde Q^{1,1}
-\tilde Q^{0,2}
+\tilde q\tilde Q^{1,0}
+\tilde Q^{0,1}\tilde q
+\tilde q\tilde Q^{0,1}
)q, \\
\partial_{t_1^{(1)}}^3q
={}&Q^{3,0}+\tilde qQ^{2,0}
+(3\tilde Q^{1,0}-2\tilde Q^{0,1}
+3\tilde q^2)Q^{1,0}\\
&+(3\tilde Q^{2,0}-3\tilde Q^{1,1}+\tilde Q^{0,2}
+6\tilde Q^{1,0}\tilde q
-3\tilde Q^{0,1}\tilde q\\
&+3\tilde q\tilde Q^{1,0}
-3\tilde q\tilde Q^{0,1}
+6\tilde q^3
)q,
\end{split}
\end{equation}
\begin{equation}
\begin{split}
\partial_{t_3^{(2)}}q
={}&-Q^{0,3}+Q^{0,2}\tilde r+Q^{0,1}\tilde R^{1,0}+q\tilde R^{2,0}, \\
\partial_{t_1^{(2)}t_2^{(2)}}^2q
={}&Q^{0,3}-Q^{0,2}\tilde r+Q^{0,1}(-\tilde R^{0,1}+\tilde r^2)\\
&+q(
\tilde R^{2,0}
-\tilde R^{1,1}
-\tilde R^{0,2}
+\tilde R^{1,0}\tilde r
+\tilde r\tilde R^{1,0}
+\tilde R^{0,1}\tilde r
), \\
\partial_{t_1^{(2)}}^3q
={}&-Q^{0,3}
+Q^{0,2}\tilde r+Q^{0,1}(
-2\tilde R^{1,0}
+3\tilde R^{0,1}
-3\tilde r^2)\\
&+q(
\tilde R^{2,0}
-3\tilde R^{1,1}
+3\tilde R^{0,2}
+3\tilde R^{1,0}\tilde r
+3\tilde r\tilde R^{1,0} \\
&-3\tilde R^{0,1}\tilde r
-6\tilde r\tilde R^{0,1}
+6\tilde r^3),
\end{split}
\end{equation}
\begin{equation}
\begin{split}
\partial_{t_3^{(1)}}r
={}&-R^{0,3}+R^{0,2}\tilde q+R^{0,1}\tilde Q^{1,0}+r\tilde Q^{2,0}, \\
\partial_{t_1^{(1)}t_2^{(1)}}^2r
={}&R^{0,3}-R^{0,2}\tilde q+R^{0,1}(
-\tilde Q^{0,1}
+\tilde q^2
)\\
&+r(
\tilde Q^{2,0}
-\tilde Q^{1,1}
-\tilde Q^{0,2}
+\tilde Q^{1,0}\tilde q
+\tilde q\tilde Q^{1,0}
+\tilde Q^{0,1}\tilde q
), \\
\partial_{t_1^{(1)}}^3r
={}&-R^{0,3}+R^{0,2}\tilde q+R^{0,1}(
-2\tilde Q^{1,0}
+3\tilde Q^{0,1}
-3\tilde q^2)\\
&+r(
\tilde Q^{2,0}
-3\tilde Q^{1,1}
+3\tilde Q^{0,2}
+3\tilde Q^{1,0}\tilde q
+3\tilde q\tilde Q^{1,0}\\
&-3\tilde Q^{0,1}\tilde q
-6\tilde q\tilde Q^{0,1}
+6\tilde q^3),
\end{split}
\end{equation}
\begin{equation}
\begin{split}
\partial_{t_3^{(2)}}r
={}&R^{3,0}
+\tilde rR^{2,0}
+\tilde R^{0,1}R^{1,0}
+\tilde R^{0,2}r, \\
\partial_{t_1^{(2)}t_2^{(2)}}^2r
={}&R^{3,0}+\tilde rR^{2,0}+(
\tilde R^{1,0}
+\tilde r^2
)R^{1,0}\\
&+(
\tilde R^{2,0}
+\tilde R^{1,1}
-\tilde R^{0,2}
+\tilde r\tilde R^{1,0}
+\tilde R^{0,1}\tilde r
+\tilde r\tilde R^{0,1})r, \\
\partial_{t_1^{(2)}}^3r
={}&R^{3,0}
+\tilde rR^{2,0}+(
3\tilde R^{1,0}
-2\tilde R^{0,1}
+3\tilde r^2)R^{1,0}\\
&+(
3\tilde R^{2,0}
-3\tilde R^{1,1}
+\tilde R^{0,2}
+6\tilde R^{1,0}\tilde r
-3\tilde R^{0,1}\tilde r\\
&+3\tilde r\tilde R^{1,0}
-3\tilde r\tilde R^{0,1}
+6\tilde r^3
)r.
\end{split}
\end{equation}


\begin{thebibliography}{sl}   

\bibitem{dimakis01}
A. Dimakis and F. Müller-Hoissen. 
Noncommutative NLS equation. 
\emph{Czech. J. Phys.}, 51 (2001), 1285-1290.

\bibitem{dimakis04}
A. Dimakis and F. Müller-Hoissen. 
Explorations of the extended ncKP hierarchy.
\emph{J. Phys. A: Math. Gen.}, 37 (2004), 10899-10930.

\bibitem{dimakis05}
A. Dimakis and F. Müller-Hoissen. 
An algebraic scheme associated with the non-commutative KP hierarchy and some of its extensions. 
\emph{J. Phys. A: Math. Gen.}, 38 (2005), 5453-5505.

\bibitem{dimakis07}
A. Dimakis and F. Müller-Hoissen. 
With a Cole-Hopf transformation to solutions of the noncommutative KP hierarchy in terms of Wronski matrices. 
\emph{J. Phys. A: Math. Theor.}, 40 (2007), F321-F329.

\bibitem{doliwa22}
A. Doliwa. 
Non-commutative Hermite-Padé approximation and integrability. 
\emph{Lett. Math. Phys.}, 112 (2022), 68.

\bibitem{doliwa22_}
A. Doliwa. 
Non-autonomous multidimensional Toda system and multiple interpolation problem. \emph{J. Phys. A: Math. Theor.}, 55 (2022), 505202.

\bibitem{doliwa24}
A. Doliwa. 
Determinantal approach to multiple orthogonal polynomials and the corresponding
integrable equations. 
\emph{Stud. Appl. Math.}, 153 (2024), e12726.

\bibitem{etingof98}
P. Etingof, I. Gelfand and V. Retakh. 
Nonabelian integrable systems, quasideterminants, and Marchenko lemma. 
\emph{Math. Res. Lett.}, 5 (1998), 1-12.

\bibitem{freeman90}
N. Freeman, C. Gilson and J. Nimmo. 
Two-component KP hierarchy and the classical Boussinesq equation. 
\emph{J. Phys. A: Math. Gen.}, 23 (1990), 4793-4804.

\bibitem{gelfand91} 
I. Gelfand and S. Retakh. 
Determinants of matrices over noncommutative rings. 
\emph{Funct. Anal. Its Appl.}, 25 (1991), 91-102. 

\bibitem{gelfand92}
I. Gelfand and S. Retakh. 
A theory of noncommutative determinants and characteristic functions of graphs. 
\emph{Funct. Anal. Its Appl.}, 26 (1992), 231-246.

\bibitem{gelfand05} 
I. Gelfand, S. Gelfand, V. Retakh and R. Wilson. 
Quasideterminants. 
\emph{Adv. Math.}, 193 (2005), 56-141.

\bibitem{gilson07}
C. Gilson and J. Nimmo. 
On a direct approach to quasideterminant solutions of a noncommutative KP equation. 
\emph{J. Phys. A: Math. Theor.}, 40 (2007), 3839-3850.

\bibitem{gilson09}
C. Gilson and S. Macfarlane. 
Dromion solutions of noncommutative Davey-Stewartson equations. 
\emph{J. Phys. A: Math. Theor.}, 42 (2009), 235202.

\bibitem{hamanaka032}
M. Hamanaka and K. Toda.
Noncommutative Burgers equation.
\emph{J. Phys. A: Math. Gen.}, 36 (2003), 11981-11998.

\bibitem{hamanaka03}
M. Hamanaka and K. Toda.
Towards noncommutative integrable systems.
\emph{Phys. Lett. A}, 316 (2003), 77-83.

\bibitem{he11}
J. He, J. Tu, X. Li and L. Wang. 
Explicit flow equations and recursion operator of the ncKP hierarchy. \emph{Nonlinearity}, 24 (2011), 2875-2890.

\bibitem{jimbo83}
M. Jimbo and T. Miwa. 
Solitons and infinite dimensional Lie algebras. 
\emph{Publ. RIMS Kyoto Univ.}, 19 (1983), 943-1001.

\bibitem{kupershmidt00}
B. Kupershmidt. 
\emph{KP or mKP: noncommutative mathematics of Lagrangian, Hamiltonian, and
integrable systems}. (American Mathematical Society, 2000).

\bibitem{li18}
C. Li. 
Symmetries and Reductions on the noncommutative Kadomtsev-Petviashvili and Gelfand-Dickey hierarchies. 
\emph{J. Math. Phys.}, 59 (2018), 123503.

\bibitem{li26}
S. Li, S. Shen, G. Yu and J. Zhang. 
A direct method in noncommutative integrable systems. 
\emph{J. Math. Phys.}, 67 (2026), 023505.

\bibitem{ma10}
W. Ma. 
An extended noncommutative KP hierarchy. 
\emph{J. Math. Phys.}, 51 (2010), 073505.

\bibitem{marchenko88}
V. Marchenko.
\emph{Nonlinear Equations and Operator Algebras}. Mathematics and Its Applications (Soviet Series), Vol. 17, D. Reidel Publishing Company, Dordrecht, 1988.

\bibitem{sakakibara04}
M. Sakakibara.  
Factorization methods for noncommutative KP and Toda hierarchy. 
\emph{J. Phys. A: Math. Gen.}, 37 (2004), L599-L604.

\end{thebibliography}
\end{document}